\documentclass[12pt]{article}
\usepackage{amsmath, amsthm, amssymb}
\usepackage{graphicx}
\usepackage{enumerate}
\usepackage{natbib}
\usepackage{url} 
\usepackage{bm}
\usepackage{mathtools} 
\usepackage{booktabs} 
\usepackage{color}
\usepackage{fge}
\usepackage{subcaption}

\newcommand{\mysetminus}{\mathbin{\vcenter{\hbox{$\scriptstyle\fgebackslash$}}}}

\usepackage{xr-hyper}
\newtheorem{proposition}{Proposition}
\newtheorem{definition}{Definition}
\newtheorem{corollary}{Corollary}
\newcommand{\blind}{0}

\begin{document}

\def\spacingset#1{\renewcommand{\baselinestretch}%
{#1}\small\normalsize} \spacingset{1}


\if0\blind
{
  \title{\bf Dependent Modeling of Temporal Sequences of Random Partitions}
  \author{Garritt L. Page\thanks{The first author gratefully acknowledges support from the Basque
Government through the BERC 2018-2021 program, by the Spanish Ministry of
Science, Innovation and Universities through BCAM Severo Ochoa
accreditation SEV-2017-0718.  The second author is supported by the grant FONDECYT 1180034 and by
ANID - Millennium Science Initiative Program - NCN17\_059.}\hspace{.2cm}\\
    Department of Statistics, Brigham Young University\\
    BCAM - Basque Center of Applied Mathematics, Bilbao, Spain, \vspace{0.1cm}\\
    Fernando A. Quintana \\
    Departamento de Estad\'istica, Pontificia Universidad Cat\'olica de Chile\\
    Millennium Nucleus Center for the Discovery of Structures in Complex Data, \vspace{0.1cm}\\
     David B. Dahl \\
    Department of Statistics, Brigham Young University}
  \maketitle
} \fi

\if1\blind
{
  \bigskip
  \bigskip
  \bigskip
  \begin{center}
    {\LARGE\bf Dependent Random Partition Models}
\end{center}
  \medskip
} \fi

\vspace{-3ex}
\begin{abstract}
We consider modeling a dependent sequence of random partitions.  It is
well-known in Bayesian nonparametrics that a random measure of discrete
type  induces a distribution over random partitions. The community has
therefore assumed that the best approach to obtain a dependent sequence of
random partitions is through modeling dependent random measures. We argue
that this approach is problematic and show that the random partition model
induced by dependent Bayesian nonparametric priors exhibits
counter-intuitive dependence among partitions even though the dependence
for the sequence of random probability measures is intuitive. Because of
this, we suggest directly modeling the sequence of random partitions when
clustering is of principal interest. To this end, we develop a class of
dependent random partition models that explicitly models dependence in a
sequence of partitions. We derive conditional and marginal properties of
the  joint partition model and devise computational strategies when
employing the method in Bayesian modeling. In the case of temporal
dependence, we demonstrate through simulation how the methodology produces
partitions that evolve gently and naturally over time. We further
illustrate the utility of the method by applying it to an environmental
data set that exhibits spatio-temporal dependence.
\end{abstract}

\noindent%

\spacingset{1.5} 
\section{Introduction}\label{sect:intro}

We introduce a method to directly model dependence in a sequence of random
partitions.  Our approach is motivated by the practical problem of
defining a prior distribution to model a sequence of random partitions
that potentially exhibits substantial concordance over time (e.g., gently
evolving clusterings over time). Traditionally, dependencies in random
partitions (i.e., the clustering of units) have been obtained as a
by-product of dependent random measures in Bayesian nonparametric (BNP)
methods. We argue, however, that when a sequence of partitions \emph{is}
the inferential object of interest, then the sequence of partitions should
be modeled \emph{directly} rather than relying  on \emph{induced} random
partition models, such as those implied by temporally dependent BNP
models. But first, we review the literature on dependent BNP methods.

\spacingset{1}
\begin{figure}[t]
\begin{center}
\includegraphics[scale=0.55]{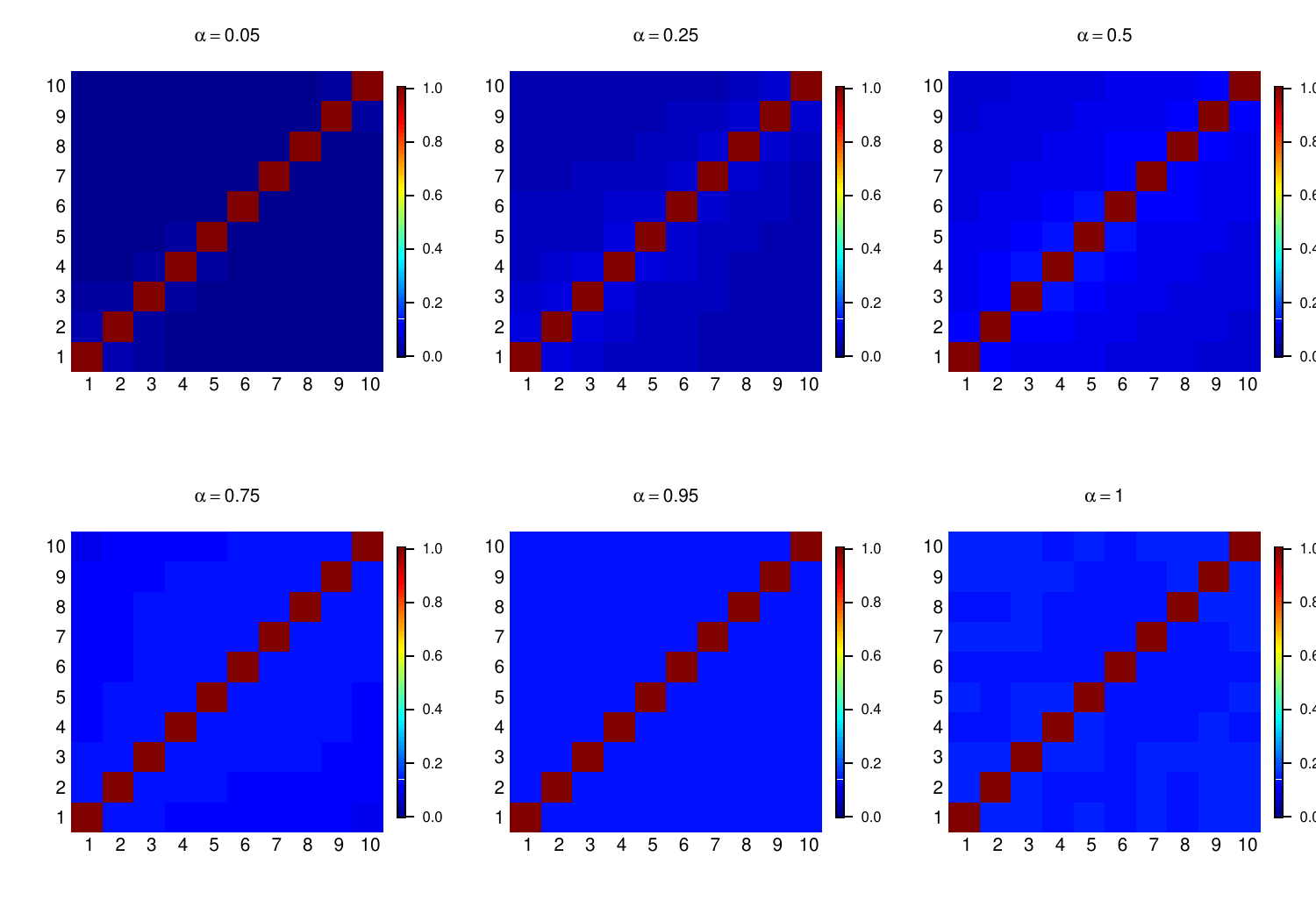}
\caption{For various values of the temporal dependence parameter $\alpha$, these
plots show the lagged ARI values using the method of \citet{caron:2017}
based on concentration parameter $M=0.5$, discount parameter set to zero,
and 10,000 Monte Carlo samples.} \label{caron2}
\end{center}
\end{figure}
\spacingset{1.5}

\spacingset{1}
\begin{figure}[ht]
\begin{center}
\hspace*{-1 cm}\includegraphics[scale=0.55]{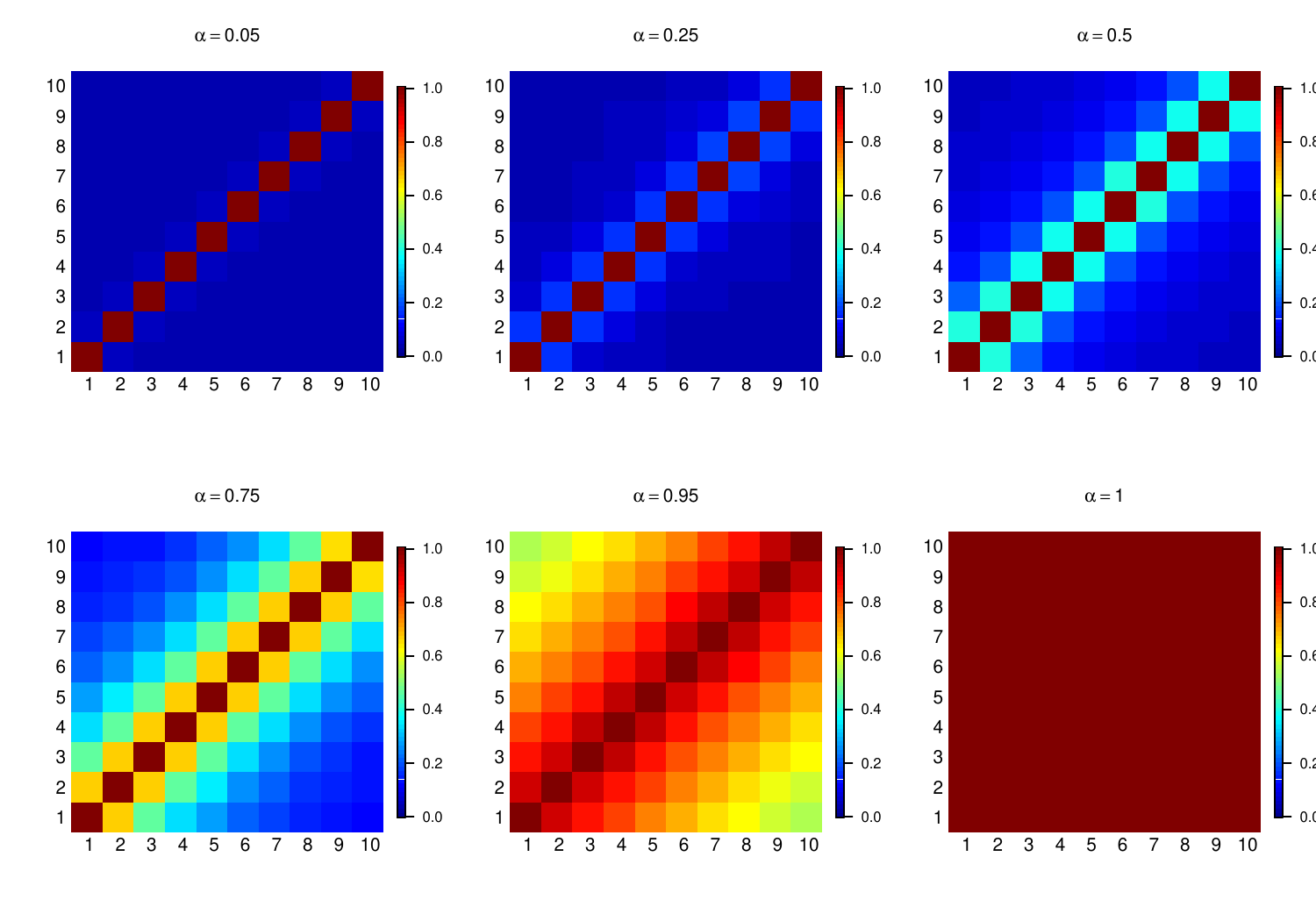}
\caption{Lagged ARI values based on 10,000 Monte Carlo samples using the method
developed in this paper.  Partitions show natural and intuitive temporal
dependence as lagged time increases and as the temporal dependence
parameter $\alpha$ increases.} \label{ours2}
\end{center}
\end{figure}
\spacingset{1.5}

A non-exhaustive list of Bayesian nonparametric methods that temporally
correlate a sequence of random probability measures include
\citet{nieto-barajas:2012}, \citet{isadora}, \citet{gutierrez:2016},
\citet{jo2017},  \cite{kalli&griffin}, \cite{DeYoreo&Kottas:18}, and
\cite{deiorio2019bayesian}. A common aspect of all these methods is that
temporal dependence is accommodated in the sequence of random measures by
way of the atoms or weights of the stick-breaking representation
\citep{Sethuraman:1994}. An alternative approach to producing a sequence
of temporally correlated random probability measures can be found in
\citet{Caron:2007} and \cite{caron:2017}. Their construction is based on a
generalised P\'olya urn scheme where dependencies between distributions
that evolve over time are induced by urn-like operations on counts and the
parameters to which they are associated. A key insight associated with all
mentioned approaches, however, is that the induced random partitions only
exhibit weak dependence even when a sequence of random probability
measures is highly correlated.  To illustrate this point, we conducted a
small Monte Carlo simulation where an induced sequence of partitions was
generated with 10 time points and 20 units using the method of
\citet{caron:2017}. To measure similarity of partitions at different time
points, we use a time-lagged adjusted Rand index (ARI)
(\citealt{hubert&arabie:1985}). Figure \ref{caron2} shows these values
averaged over 10,000 Monte Carlo samples. Notice that as the temporal
dependence parameter ($\alpha$) increases, the partitions from time period
$t$ to $t+1$ only become slightly more similar, such that the dependence
between partitions is, at best, only weak.  Further, the dependence is not
temporally intuitive as it does not decay as a function of lagged time. In
contrast, compare the dependence structure in Figure \ref{caron2} to that
in Figure \ref{ours2}, which contains average lagged ARI values between
time lagged partitions generated using the method developed in this paper.
Notice that unlike the induced random partitions generated using the
method in \citet{caron:2017}, the sequence of partitions generated using
our approach displays intuitive temporal dependence.  That is, as the time
lag increases, similarity between partitions decreases and as the temporal
dependence parameter $\alpha$ increases, the partitions become dissimilar
over time at a slower rate.

The counter-intuitive behavior displayed in Figure \ref{caron2} is not
unique to the approach of \cite{caron:2017}. As noted by
\cite{wade&walker&petrone:2014}, the same type of behavior is present when
using a linear dependent Dirichlet process mixture model.  In fact,
all BNP methods that model a sequence of discrete random probability
measures will induce a random partition model with similarly weak
correlation behavior. This behavior is analogous to trying to induce
dependence among random variables from distributions with correlated
parameters.  There is no guarantee that correlated parameters would
produce strong correlations among the random variables themselves.

Our approach is to consider the sequence of partitions as the parameter of
principal interest and develop a method that models it directly. This will
provide more control over how ``smoothly'' partitions evolve over time.
Perhaps the work closest to ours, in the sense of explicitly modeling a
sequence of partitions, can be found in \cite{carlosetal}. Their modeling
approach for a temporally-referenced sequence of partitions  can be
applied to only two time points and differs from ours in that they do not
focus on smooth evolution of partitions over time.

The remainder of the article is organized as follows. In
Section~\ref{model} we present the proposed approach for a sequence of
dependent random partitions, discuss its main properties, and suitable
computational strategies for inference based on posterior simulation.
Section \ref{sec:simulations.studies} contains the results from three
simulation studies that further explore aspects of the model.
Section~\ref{applications} describes an environmental data application and
some concluding remarks are provided in Section~\ref{conclusions}.  An
accompanying Supplementary Materials file collects the proofs of results
stated below, provides details on posterior simulation algorithms, and
contains further simulation and data analysis results.

\section{Joint Model for a Sequence of Partitions}\label{model}

Before detailing our method, we introduce some general notation. Let $i=1,
\ldots, m$ denote the $m$ experimental units at time $t$ for $t=1,\ldots,
T$. Let $\rho_t = \{S_{1t}, \ldots, S_{k_tt}\}$ denote a partition of the
$m$ experimental units at time $t = 1, \ldots, T$ into $k_t$ clusters. An
alternative notation is based on $m$ cluster labels at time $t$ denoted by
$\bm{c}_t  = \{c_{1t}, \ldots, c_{mt}\}$ where $c_{it} = j$
implies that $i \in S_{jt}$. 
Finally, any quantity with a ``$\star$'' superscript will be
cluster-specific. For example, we will use $\mu^{\star}_{jt}$ to denote
the mean of cluster $j$ at time $t$  so that $\mu_{it} =
\mu^{\star}_{c_{it}t}$. 

\subsection{Temporal Modeling for Sequences of Partitions} \label{dependentPartitions}

Introducing temporal dependence in a collection of  partitions  requires
formulating a joint probability model for $(\rho_1, \ldots, \rho_T)$.
Generically, we will denote this joint model with $\text{Pr}(\rho_t,
\ldots, \rho_T)$.   Temporal dependence among the $\rho_t$'s implies that
the cluster configuration in $\rho_t$ could be impacted by that found in  $\rho_{t-1}, \rho_{t-2}, \ldots,
\rho_1$.  However, we
assume that  the probability model for the sequence of partitions has a
first-order Markovian structure.  That is, the conditional distribution of
$\rho_t$ given $\rho_{t-1}, \rho_{t-2}, \ldots, \rho_1$ only depends on
$\rho_{t-1}$.  Thus, we construct $\text{Pr}(\rho_t, \ldots, \rho_T)$ as
\begin{align} \label{joint.model}
\text{Pr}(\rho_1, \ldots, \rho_T) = \text{Pr}(\rho_T \mid  \rho_{T-1}) \text{Pr}(\rho_{T-1} \mid  \rho_{T-2})\cdots \text{Pr}(\rho_2 \mid  \rho_1)\text{Pr}(\rho_1).
\end{align}
Here $\text{Pr}(\rho_1)$ is an exchangeable partition probability function
(EPPF) that describes how the $m$ experimental units at time period 1 are
grouped into $k_1$ distinct groups with frequencies  $n_{11}, \ldots,
n_{1k_1}$.  One characteristic of an EPPF that will prove useful in what
follows is sample size consistency, or what
\cite{deblasi&favaro&lijoi&mena&prunster&ruggiero:2015} refer to as the
{\em addition rule}. This property dictates that marginalizing the last of
$m + 1$ elements leads to the same model as if we only had $m$ elements. A
commonly encountered EPPF is that induced by a Dirichlet process (DP).
This particular EPPF is sometimes referred to as a Chinese restaurant
process (CRP) which can be seen as  a special case from the family
of product partition models (PPM). For more details, see
\cite{deblasi&favaro&lijoi&mena&prunster&ruggiero:2015}. Because we employ
the EPPF of the CRP in what follows, we provide its form here
\begin{align}\label{crp}
\text{Pr}(\rho \mid  M) = \frac{M^k}{\prod_{i=1}^{n}(M+i-1)} \prod_{i=1}^k (|S_i| - 1)!,
\end{align}
where $k$ is the number of clusters in $\rho$ and $M$ is a concentration
parameter controlling the number of clusters.  We will denote this random
partition distribution as $CRP(M)$.

Although conceptually straightforward, \eqref{joint.model} is silent
regarding how $\rho_{t-1}$ influences the form of $\rho_{t}$.  To make
this explicit, we introduce an auxiliary variable that guides the
similarity between $\rho_t$ and $\rho_{t-1}$. Note that if two partitions
are highly dependent, then the cluster configurations between them will
change very little and as a result only a few of the $m$ experimental
units will change cluster assignment. Conversely, two partitions that
exhibit low dependence will likely be comprised of very different cluster
configurations. The auxiliary variable we introduce identifies which of
the experimental units at time $t-1$ will be considered for possible
cluster reallocation at time $t$. Specifically, let $\gamma_{i t}$ denote
the following
\begin{align}
\gamma_{i t} =
\left\{
\begin{array}{c l}
 1 & \mbox{if unit $i$ is {\it not} reallocated when moving from time  $t-1$ to $t$}   \\
 0 &  \mbox{otherwise},
\end{array}
\right.
\end{align}
for $i=1, \ldots,m$.  Notice that when $\gamma_{it} = 0$, item $i$ is
subject to reallocation at time $t$, but still may, by random assignment,
end up in the same cluster at time $t$ as it was at time $t-1$. By
construction, we set $\gamma_{i1} = 0$ for all $i$, i.e., all experimental
units are allocated to clusters during the first time period. We then
assume that $\gamma_{i t} \stackrel{ind}{\sim} Ber(\alpha_t)$. Note that
each of the $\alpha_t \in [0,1]$ acts as a temporal dependence parameter.
Specifically, we will interpret $\alpha_t =1$ as implying that $\rho_t =
\rho_{t-1}$ with probability 1. Conversely, when $\alpha_t=0$, then
$\rho_t$ is independent of $\rho_{t-1}$. Further, when $\alpha_t$ is
constant for all $t$, the degree of dependence among partitions is
constant over time, whereas general values for $\alpha_t$ provide for
varying degrees of dependence and more flexible partition patterns over
time. For notational convenience, we introduce $\bm{\gamma}_{t} =
(\gamma_{1 t}, \gamma_{2 t}, \ldots, \gamma_{m t})$ which is an $m$-tuple
comprised of zeros and ones. The augmented joint model changes
\eqref{joint.model} to
\begin{multline} \label{joint.joint.model}
\text{Pr}(\bm{\gamma}_1,\rho_1, \ldots, \bm{\gamma}_T, \rho_T) =  \text{Pr}(\rho_T \mid  \bm{\gamma}_T, \rho_{T-1}) \text{Pr}(\bm{\gamma}_T) \times \\
\text{Pr}(\rho_{T-1} \mid \bm{\gamma}_{T-1},
\rho_{T-2})\text{Pr}(\bm{\gamma}_{T-1})\cdots \text{Pr}(\rho_2 \mid
\bm{\gamma}_2, \rho_1)\text{Pr}(\bm{\gamma}_2) \text{Pr}(\rho_1).
\end{multline}
We describe $\text{Pr}(\rho_t \mid  \bm{\gamma}_t, \rho_{t-1})$ shortly,
but first provide a definition.

\begin{definition}
We say that partitions $\rho_{t-1}$ and $\rho_t$ are compatible with
respect to $\bm{\gamma}_t$, if $\rho_t$ may be obtained from $\rho_{t-1}$
by reallocating items as indicated by $\bm{\gamma}_t$, i.e., those items
$i$ such that $\gamma_{ti} = 0$ for $i=1,\ldots,m$.  Note that the compatibility relation is an equivalence relation.

\end{definition}

There is a simple way to check if $\rho_{t-1}$ is compatible with $\rho_t$
with respect to $\bm{\gamma}_t$. Let $\mathfrak{R}_{t} = \{i : \gamma_{it}
= 1\}$ be the collection of units that  remain fixed when moving from time
$t-1$ to time $t$, and $\mathfrak{R}^C_{t}= \{i : \gamma_{it} = 0\}$ is
the collection of units that do not. Next denote with
$\rho^{\mathfrak{R}_t}_t$ the ``reduced'' partition at time $t$  that
remains after removing all items in $\mathfrak{R}^C_{t}$ from the subsets
of $\rho_t$.
Similarly, let $\rho^{\mathfrak{R}_t}_{t-1}$ be the reduced partition at
time $t-1$ based on $\bm{\gamma}_t$. Then $\rho_{t-1}$ and $\rho_t$ are
compatible with respect to $\bm{\gamma}_t$ if and only if
$\rho^{\mathfrak{R}_t}_{t-1} = \rho^{\mathfrak{R}_t}_t$.

Now, to further characterize $\text{Pr}(\rho_t \mid  \bm{\gamma}_t,
\rho_{t-1})$, let $P$ denote the set of all partitions of $m$ units and
let $P_{C_t} = \{\rho_t \in P : \rho^{\mathfrak{R}_t}_{t-1} =
\rho^{\mathfrak{R}_t}_t \}$ be the collection of partitions at time $t$
that are compatible with $\rho_{t-1}$ based on $\bm{\gamma}_t$.  Then, by
construction, $\text{Pr}(\rho_t \mid  \bm{\gamma}_t, \rho_{t-1})$ is a
random partition distribution whose support is $P_{C_t}$  so that
\begin{align*}
\text{Pr}(\rho_t = \lambda \mid  \bm{\gamma}_t, \rho_{t-1}) = \displaystyle \frac{\text{Pr}(\rho_t = \lambda) \text{I}[\lambda \in P_{C_t}]}{\sum_{\lambda^\prime} \text{Pr}(\rho_t = \lambda^\prime)\text{I}[\lambda^\prime \in P_{C_t}]},
\end{align*}
where  $\text{Pr}(\rho_t = \lambda)$ is
the EPPF at the first time point evaluated at $\lambda$. Here, and in what
follows, $\text{I}[A]$ denotes an indicator function with $\text{I}[A]=1$
if statement $A$ is true, and $0$ otherwise.

It would be appealing if marginally each of the $\rho_t$ follow the
assumed EPPF for $\rho_1$, so that the joint probability model for
partitions would be stationary. The following proposition establishes this
result, which is a consequence of the fact that conditioning on
$\bm{\gamma}_t$ provides a ``reduced'' EPPF.
\begin{proposition}  \label{proposition1}
Let $\rho_1 \sim EPPF$  and $\bm{\gamma}_1 = \bm{0}$.  If a joint model for
$\rho_1 \ldots, \rho_T$ is constructed as described above by introducing
$\bm{\gamma}_t$ for $t = 2, \ldots, T$, then we have that marginally
$\rho_{1},\ldots,\rho_{T}$ are identically distributed with law coming
from the EPPF used to model $\rho_1$.  Specifically, letting  $\rho_{-t} =
(\rho_1, \ldots, \rho_{t-1}, \rho_{t+1}, \ldots, \rho_T)$ and $\bm{\gamma}
= (\gamma_1, \ldots, \gamma_T)$, we have that for all $\lambda \in P$,

\begin{align*}
\textup{Pr}(\rho_t = \lambda) = \sum_{\rho_{-t} \in P^{\otimes m}} \sum_{\bm{\gamma} \in \Gamma^{\otimes m}} \textup{Pr}(\bm{\gamma}_1, \rho_1, \ldots, \rho_t=\lambda, \ldots, \bm{\gamma}_T, \rho_T) = \textup{Pr}(\rho_1=\lambda),
\end{align*}
where $P^{\otimes m} = P \times P \times \ldots \times P$,  $P$ a
collection of all partitions of $m$ units, $\Gamma^{\otimes m} = \Gamma
\times \Gamma \times \ldots \times \Gamma$, and $\Gamma$ a collection of
all possible binary vectors of size $m$.
\end{proposition}
\begin{proof}
See supplementary material.
\end{proof}

In what follows, we will use $tRPM(\bm{\alpha}, M)$ to denote our temporal
random partition model \eqref{joint.joint.model} parameterized by
$\alpha_1, \ldots, \alpha_T$ and the EPPF in \eqref{crp}. We briefly
mention that introducing $\gamma_{it}$ is similar in spirit to the
approach taken by \cite{Caron:2007, caron:2017}. However, they use
$\bm{\gamma}_t$ to identify a partial partition at time  $t$ that informs
how {\it all} the observational units will be reallocated at time  $t+1$.
While this difference may seem inconsequential at first glance, it has
drastic ramifications on the type of dependence that exists among the
actual sequence of partitions. This is illustrated in Figures \ref{caron2}
and \ref{ours2} provided in the Introduction.  The sequence of partitions
used to create Figure  \ref{ours2} were generated using the
$tRPM(\bm{\alpha}, M)$ with $M=0.5$. As mentioned in the Introduction,
when the main emphasis is on modeling a ``smoothly'' evolving sequence of
random partitions, the temporal dependence displayed in Figure \ref{ours2}
is much more natural than that found in Figure \ref{caron2}.

\subsection{Dependence in Partitions}
We now further explore how our method models dependence across partitions.
To do this, we analyze closeness between partitions $\rho_1$ and $\rho_2$
by way of co-clustering  of cluster labels $(c_{11},\ldots,c_{m1})$ and
$(c_{12},\ldots,c_{m2})$, respectively.  We base our exploration on the
Rand index which is defined a
$$R(\rho_1,\rho_2)=\frac{a+b}{\binom{m}{2}},$$
where $a$ is the number of pairs $(i,j)$ with $i,j\in [m]=\{1,\ldots, m\}$ that simultaneously co-cluster in $\rho_1$ and
$\rho_2$ and $b$ is the number of such pairs that simultaneously do not co-cluster. Writing
$\varphi_{ij}=P(c_{i1}=c_{j1},c_{i2}=c_{j2}) + P(c_{i1}\ne c_{j1},c_{i2}\ne c_{j2})$,
we note that
$$E\left[R(\rho_1,\rho_2) \right]=\binom{m}{2}^{-1}\sum_{1\le i<j\le m} \varphi_{ij}.$$
To provide context to the co-clustering probabilities of cluster labels based on $tRPM(\alpha,M)$, we consider the model proposed in~\cite{Caron:2007}. 
In their approach and assuming $\rho_1\sim CRP(M)$, each $i\in [m]$ is
randomly removed from the partition with probability $1-\alpha$, and
$\rho_2$ is formed by running an extra $CRP(M)$ process, but starting from
an urn that has weights given by the normalized cluster sizes left from
the removal process. See details in~\cite{Caron:2007}. We denote
partitions that follow the model of \cite{Caron:2007} as
$\rho_1,\rho_2\sim CAR(\alpha,M)$. For both our method and
$CAR(\alpha,M)$, the case that $\alpha=0$ leads to $\rho_1,\rho_2
\stackrel{iid}{\sim}CRP(M)$, while the largest degree of dependence
between $\rho_1$ and $\rho_2$ is achieved when $\alpha=1$.  The following
proposition characterizes the co-clustering probabilities under our method
and $CAR(\alpha,M)$.
\begin{proposition}\label{prop:corr}
Let $m=T=2$, so that $E\left[R(\rho_1,\rho_2) \right] = \varphi_{12}$.
\begin{itemize}
\item[(a)] If $\rho_1,\rho_2\sim tRPM(\alpha,M)$, where to simplify
    notation we write $\alpha\equiv \alpha_2$, then
$$\varphi_{12}=\alpha^2+\frac{(1+M^2)}{(1+M)^2}(1-\alpha)^2.$$
\item[(b)] If $\rho_1,\rho_2\sim CAR(\alpha,M)$ then
$$\varphi_{12}=
\left[\frac{6+3M+4M^2+M^3)}{(M+1)(M+2)(M+3))}\right] \alpha^{2} +
\frac{(1+M^{2})}{(1+M)^{2}} \left(1-\alpha^{2}\right).$$
\end{itemize}
\end{proposition}
\begin{proof}
See supplementary material.
\end{proof}

An interesting consequence of Proposition~\ref{prop:corr} is that we can
compute the expected value of the Rand index in the case $\rho_1,\rho_2
\stackrel{iid}{\sim}CRP(M)$.

\begin{corollary}\label{prop:coro}
If $\rho_1,\rho_2 \stackrel{iid}{\sim}CRP(M)$ then for any $m\ge 2$,
$$E\left[R(\rho_1,\rho_2) \right]=\frac{(1+M^{2})}{(1+M)^{2}}.$$
\end{corollary}
\begin{proof}
The result follows immediately by noting that the i.i.d. case coincides
with $tRPM(0,M)$ and that by exchangeability and independence,
$\varphi_{ij}=\varphi_{12}$ for all $1\le i<j\le m$.
\end{proof}

The result from Proposition~\ref{prop:corr} (a) shows that, under the
$tRPM(\alpha,M)$ model, $\lim_{\alpha\rightarrow
0^+}\varphi_{12}=\frac{(1+M^{2})}{(1+M)^{2}}$, i.e., it agrees with
$E\left[R(\rho_1,\rho_2) \right]$ under the i.i.d. case. The same holds as
$\alpha\rightarrow 0^+$ under the $CAR(\alpha,M)$ model. Furthermore, for
the $tRPM(\alpha,M)$, we get the appealing result that
$\lim_{\alpha\rightarrow 1^+}\varphi_{12}=1$, but the same limit under the
$CAR(\alpha,M)$ is a number strictly less than $1$ for any $M>0$. This
reveals that the closeness between partitions under the proposed
$tRPM(\alpha,M)$, as measured by the $\varphi_{12}$ quantity, can attain
its maximum value of 1, which simply corresponds to the case where
none of the units is relocated.  The same cannot hold for the
$CAR(\alpha,M)$ model because partitions are linked through a latent
mechanism rather than directly as in the proposed model. Finally, we
conjecture that similar results can be obtained for $m>2$ but calculations
become more involved. The result from Corollary~\ref{prop:coro} is
nevertheless valid for any $m\ge 2$.

\subsection{Toy Example to Illustrate Conditional Model} \label{toy.example}

To build intuition regarding the transition from $\rho_{t-1}$ to $\rho_t$,
consider the conditional probabilities in equation
\eqref{joint.joint.model} and the very simple scenario of $m=3$ and $T=2$.
We have that
\begin{align*}
\text{Pr}(\rho_2 \mid  \rho_{1}) =  \sum_{\bm{\gamma}_{2} \in \Gamma}  \text{Pr}(\rho_2\mid   \bm{\gamma}_{2}, \rho_{1})\text{Pr}(\bm{\gamma}_{2}),
\end{align*}
where again, $\Gamma$ is the collection of all possible binary $3$-tuples
and operate under $\rho_1 \sim CRP(M)$. The conditional probabilities are
provided in Table \ref{ConditionalCRP}, where we set $M=1$ for simplicity.
From Table \ref{ConditionalCRP}, notice that $\text{Pr}(\rho_2 \mid
\rho_{1})$ is a reweighted CRP and that, as $\alpha \rightarrow 0$,
partition probabilities correspond to those from the original CRP and, as
$\alpha \rightarrow 1$,  $\text{Pr}(\rho_2 = \rho_1) \rightarrow  1$.
Further, notice that partitions associated with $\rho_2$ that are more
similar to $\rho_1$ are given larger weight relative to a CRP.  For
example, given $\rho_1 = \{\{1, 2, 3\}\}$ then $\rho_2 =\{ \{1, 2\},
\{3\}\}$ has higher probability than $\rho_2 = \{\{1\}, \{2\}, \{ 3\}\}$
for any $\alpha > 0$ but have equal probability in a CRP.  From this toy
example we see that the conditional co-clustering probabilities display
dependencies in line with the desire to have partitions evolve gently over
time.

\spacingset{1}
\begin{table}[t]
\caption{Partition probabilities from the conditional distribution
$\text{Pr}(\rho_2 \mid  \rho_1)$ using a CRP EPPF}

\hspace*{-0cm} \resizebox{\columnwidth}{!}{\begin{tabular}{l@{\hskip 0.01in}  c@{\hskip 0.05in} l @{\hskip 0.05in} l @{\hskip 0.05in} l@{\hskip 0.05in} l @{\hskip 0.05in} l }
\multicolumn{1}{c}{$(c_1,c_2,c_3)$}  &  $\text{Pr}(\rho_1)$  & \multicolumn{1}{c}{$\text{Pr}(\rho_2 \mid  \rho_1 = a)$}  & \multicolumn{1}{c}{$\text{Pr}(\rho_2 \mid  \rho_1 = b)$} & \multicolumn{1}{c}{$\text{Pr}(\rho_2 \mid  \rho_1 = c)$} & \multicolumn{1}{c}{$\text{Pr}(\rho_2 \mid  \rho_1 = d)$} & \multicolumn{1}{c}{$\text{Pr}(\rho_2 \mid  \rho_1 = e)$}\\ \midrule
 $a = (1,1,1)$ 	& $\frac{2}{6}$ 	& $\frac{2}{6}[1 + 3\alpha^2 - \alpha^3]$  	& $\frac{2}{6} [1- \alpha^2]$ 			& $\frac{2}{6} [1- \alpha^2]$ 				& $\frac{2}{6}  [1- \alpha^2]$ 				& $\frac{2}{6} [1-3\alpha^2 + 2\alpha^3]$ \\ [4pt]
 $b = (1,1,2)$ 	& $\frac{1}{6}$ 	& $\frac{1}{6}[1- \alpha^2 ]$ 			& $\frac{1}{6} [1+3\alpha^2 + 2\alpha^3]$  & $\frac{1}{6} [1- \alpha^2]$  				&  $\frac{1}{6} [1- \alpha^2]$				& $\frac{1}{6} [1+ \alpha^2 - 2\alpha^3]$  \\ [4pt]
 $c=  (1,2,1)$ 	& $\frac{1}{6}$ 	& $\frac{1}{6}[1- \alpha^2 ]$  			& $\frac{1}{6} [1- \alpha^2]$  			& $\frac{1}{6} [1 + 3\alpha^2 + 2\alpha^3]$  	& $\frac{1}{6}  [1- \alpha^2]$  				& $\frac{1}{6} [1+ \alpha^2 - 2\alpha^3]$  \\ [4pt]
 $d = (1,2,2)$ 	& $\frac{1}{6}$ 	& $\frac{1}{6}[1- \alpha^2 ]$  			& $\frac{1}{6} [1- \alpha^2]$   			& $\frac{1}{6} [1- \alpha^2]$   				& $\frac{1}{6}  [1 + 3\alpha^2 + 2\alpha^3]$  	& $\frac{1}{6} [1+ \alpha^2 - 2\alpha^3]$   \\ [4pt]
 $e = (1,2,3)$ 	& $\frac{1}{6}$ 	& $\frac{1}{6}[1-3\alpha^2 + 2\alpha^3]$   	&  $\frac{1}{6} [1+ \alpha^2 - 2\alpha^3]$  & $\frac{1}{6} [1+ \alpha^2 - 2\alpha^3]$    	&  $\frac{1}{6} [1+ \alpha^2 - 2\alpha^3]$    	& $\frac{1}{6} [1 + 3\alpha^2 + 2\alpha^3]$    \\ \midrule
\end{tabular}}
\label{ConditionalCRP}
\end{table}%
\spacingset{1.5}
\subsection{Hierarchical Data Model}\label{hierarchical.model}

Once a partition model is specified, there is tremendous flexibility
regarding how to model time (global or cluster-specific) at different
levels of a hierarchical model (at the data level, process level, or
both).  Since we are interested to see how including time in the partition
model impacts clustering and model fits, in the simulations of Section
\ref{sec:simulations.studies}, we consider a hierarchical model where time
only appears in the partition model.  In particular, using cluster label
notation, we will employ the following hierarchical model
\begin{equation}\label{dat.gen.mech}
  \begin{aligned}
Y_{it} \mid  \bm{\mu}^{\star}_{t}, \bm{\sigma}_t^{2\star}, \bm{c}_{t} & \stackrel{ind}{\sim} N(\mu^{\star}_{c_{it}t}, \sigma_{c_{it}t}^{2\star}), \ i = 1, \ldots, m \ \mbox{and} \ t=1, \ldots, T,   \\
(\mu_{jt}^{\star}, \sigma^{\star}_{jt})\mid \theta_t, \tau_t^2 & \stackrel{ind}{\sim} N(\theta_t, \tau_t^2) \times UN(0,A_{\sigma}), \ j = 1, \dots, k_t ,\\
(\theta_t, \tau_t) & \stackrel{iid}{\sim} N(\phi_0, \lambda^2) \times UN(0, A_{\tau}), \ t = 1, \ldots, T,\\
(\phi_0, \lambda) & \sim N(m_0, s^2_0) \times UN(0, A_{\lambda}), \\
\{\bm{c}_{t}, \ldots, \bm{c}_{T}\} & \sim tRPM(\bm{\alpha}, M), \ \mbox{with $\alpha_t \stackrel{iid}{\sim} Beta(a_{\alpha}, b_{\alpha})$},
  \end{aligned}
\end{equation}
where  $Y_{it}$ denotes the response measured on the $i$th unit  at time
$t$, $UN$ denotes a uniform distribution and $A_{\sigma}$, $A_{\tau}$,
$A_{\lambda}$, $m_0$, $s^2_0$, $a_{\alpha}$, $b_{\alpha}$,  $M$ are
user-supplied hyper-parameters.  The remaining assumptions (e.g.,
independence across clusters and exchangeability within each cluster) are
commonly employed.

\subsection{Computation} \label{PostComputation}

As the posterior distribution implied by model \eqref{dat.gen.mech} is not
of a known form, we build an algorithm to sample from it.  The
construction of $\text{Pr}(\rho_1, \ldots, \rho_T)$ naturally leads one to
consider a Gibbs sampler. In the Gibbs sampler, $\bm{\gamma}_t$ will need
to be updated in addition to $\rho_t$ (by way of $\bm{c}_t$). But the
Markovian assumption reduces some of the cost as we only need to consider
$\rho_{t-1}$ and $\rho_{t+1}$ when updating $\rho_t$. Even though each
update of $\rho_t$ and $\bm{\gamma}_t$ for $t=1,\ldots, T$ needs to be
checked for compatibility, it is fairly straightforward to adapt standard
algorithms, e.g. Algorithm 8 of \cite{neal:2000}, with care to make sure
that only experimental units with $\gamma_{it} = 0$ are considered when
updating $c_{it}$. Here we provide a general sketch for updating $c_{it}$
and $\gamma_{it}$ within an MCMC algorithm, with much more detail provided
in Section \ref{MCMC} the online supplementary material

The MCMC algorithm we employ depends on deriving the complete conditionals
for $\rho_t$ and $\gamma_t$. A key result needed to derive them  is
provided in the following proposition.
\begin{proposition}\label{proposition2}
Based on the construction of a joint probability model as described in
Section \ref{dependentPartitions}, we have
\begin{align}
\textup{Pr}(\rho_t \mid  \bm{\gamma}_t, \rho_{t-1}) = \left\{
\begin{array}{l l}
 \textup{Pr}(\rho_t)/\textup{Pr}(\rho^{\mathfrak{R}_t}_t)  &    \mbox{if  $\rho^{\mathfrak{R}_t}_{t-1} = \rho^{\mathfrak{R}_t}_t$}\\[3pt]
 0 &     otherwise.
\end{array}
\right.
\end{align}

\end{proposition}

\begin{proof}
See the supplementary material.
\end{proof}

When updating $\gamma_{it}$ in a Gibbs sampler, one can think of removing
$\gamma_{it}$ from $\bm{\gamma}_t$, and then reinsert it as either a 0 or
1. To this end let $\mathfrak{R}^{(-i)}_{t} = \mathfrak{R}_{t} \mysetminus
\{i\}$ and $\mathfrak{R}^{(+i)}_{t} = \mathfrak{R}^{(-i)}_{t} \cup \{i\}$
and let $\bm{\gamma}_{t,+i}$ denote the $\bm{\gamma}_t$ vector with the
$i$th entry set to 1.  Then the full conditional for $\gamma_{it} = 1$,
denoted by $\text{Pr}(\gamma_{it} = 1 \mid  -)$, is 
\begin{align*}
\text{Pr}(\gamma_{it} = 1 \mid  -) &  \propto \text{Pr}(\rho_t \mid \bm{\gamma}_{t,+i},  \rho_{t-1}) \text{Pr}(\bm{\gamma}_{t,+i})\text{I}[\rho^{\mathfrak{R}^{(+i)}_{t}}_{t-1} = \rho^{\mathfrak{R}^{(+i)}_{t}}_{t}], \\
                  		  & \propto \frac{\text{Pr}(\rho_t)}{\text{Pr}(\rho^{\mathfrak{R}^{(+i)}_{t} }_t)} \alpha_t^{\gamma_{it}}\text{I}[\rho^{\mathfrak{R}^{(+i)}_{t}}_{t-1} = \rho^{\mathfrak{R}^{(+i)}_{t}}_{t}],
\end{align*}
which results in
\begin{align}\label{full.conditional.gamma1}
\text{Pr}(\gamma_{it} = 1 \mid  -) = \displaystyle\frac{\alpha_t \text{Pr}(\rho^{\mathfrak{R}^{(-i)}_{t} }_t)}{\alpha_t \text{Pr}(\rho^{\mathfrak{R}^{(-i)}_{t} }_t) + (1-\alpha_t) \text{Pr}(\rho^{\mathfrak{R}^{(+i)}_{t} }_t)}\text{I}[\rho^{\mathfrak{R}^{(+i)}_{t}}_{t-1} = \rho^{\mathfrak{R}^{(+i)}_{t}}_{t}].
\end{align}
For a given EPPF that has a closed form (e.g., CRP), it is straightforward
to compute $\text{Pr}(\rho^{\mathfrak{R}^{(-i)}_{t} }_t)$ and
$\text{Pr}(\rho^{\mathfrak{R}^{(+i)}_{t} }_t)$. If, however, the EPPF does
not have a closed form, then note that \eqref{full.conditional.gamma1} can
be re-expressed as
\begin{align}\label{full.conditional.gamma1b}
\text{Pr}(\gamma_{it} = 1 \mid  -) = \displaystyle\frac{\alpha_t}{\alpha_t  + (1-\alpha_t)\text{Pr}(\rho^{\mathfrak{R}^{(+i)}_{t} }_t)/\text{Pr}(\rho^{\mathfrak{R}^{(-i)}_{t} }_t)}\text{I}[\rho^{\mathfrak{R}^{(+i)}_{t}}_{t-1} = \rho^{\mathfrak{R}^{(+i)}_{t}}_{t}].
\end{align}
The quantity $\text{Pr}(\rho^{\mathfrak{R}^{(+i)}_{t}
}_t)/\text{Pr}(\rho^{\mathfrak{R}^{(-i)}_{t} }_t)$ is a commonly
encountered expression in MCMC methods that employ Neal's  Algorithm 8
\citep{neal:2000}. Those same methods can be employed to calculate the
desired probabilities. See Section \ref{gamma.update} of the online supplementary material for more
detail.

When updating $c_{it}$ note that, within the MCMC algorithm, only those
$c_{it}$ for which $\gamma_{it} = 0$ are updated. Thus
$\rho^{\mathfrak{R}_{t}}_{t-1} = \rho^{\mathfrak{R}_{t}}_{t}$ by
construction.  As a result, only compatibility between $\rho_{t}$ and
$\rho_{t+1}$ (i.e.,  $\rho^{\mathfrak{R}_{t+1}}_{t} =
\rho^{\mathfrak{R}_{t+1}}_{t+1}$) needs to be checked when updating
$c_{it}$.  Now letting $\text{Pr}(c_{it} = h) = \text{Pr}(c_{1t}, \ldots,
c_{it} = h, \ldots, c_{mt})$ and denoting the partition based on
$\{c_{1t}, \ldots, c_{it} = h, \ldots, c_{mt}\}$  as $\rho_{t:c_{it}=h} =
\{ S_{1t}^{-i}, \ldots, S_{ht}^{-i}\cup \{i\}, \ldots, S_{k_t^{-i}
t}^{-i}\}$ where $S^{-i}_{jt}$ denotes the $j$th cluster at time $t$ with
the $i$th unit removed (note it is possible that  $S^{-i}_{jt} = S_{jt}$),
the  full conditional multinomial probability for $c_{it}$  is
\begin{align*}
\text{Pr}(c_{it} = h \mid  -) \propto
\left\{
\begin{array}{cl}
N(Y_{it} \mid  \mu^{\star}_{c_{it} = h, t}, \sigma^{2\star}_{c_{it} = h, t})\text{Pr}(c_{it} = h)\text{I}[\rho^{\mathfrak{R}_{t+1}}_{t:c_{it}=h} = \rho^{\mathfrak{R}^{t+1}}_{t+1}]  & \mbox{for $h = 1, \ldots, k_t^{-i}$, }   \\
N(Y_{it} \mid  \mu^{\star}_{new_h, t}, \sigma^{2\star}_{new_h, t})\text{Pr}(c_{it} = h)\text{I}[\rho^{\mathfrak{R}_{t+1}}_{t:c_{it}=h} = \rho^{\mathfrak{R}^{t+1}}_{t+1}] &    \mbox{for $h = k_t^{-i}+1$,}
\end{array}
\right.
\end{align*}
where  $\mu^{\star}_{new_h, t}$ along with $\sigma^{2\star}_{new_h, t}$
are auxiliary parameters drawn from the prior as in \cite{neal:2000}'s
Algorithm 8 (with one auxiliary parameter) and $k_t^{-i}$  is the number 
of clusters at time $t$ when the $i$th unit has been removed. Further $N(\cdot \mid m, s^2)$ denotes a normal density with mean $m$ and variance $s^2$. Given
$\rho_t$ and $\bm{\gamma}_t$, the full conditionals of the remaining
parameters in model \eqref{dat.gen.mech} follow standard techniques. A
sample can be drawn from the posterior distribution implied by  model
\eqref{dat.gen.mech} by iterating through the complete conditionals for
$\bm{\gamma}_t$ and $\rho_t$ and those of other model parameters. See Section \ref{rho.update} of the
online supplementary material for more detail. 

In our experience the MCMC algorithm described here and in Section \ref{MCMC} of the online supplementary material generally behaves well with regards to mixing and convergence.  However,   applications  where $\alpha \approx 1$ can negatively affect the performance of the algorithm.  Having a parameter close to the boundary of its support commonly produces computational issues.  It is possible to mitigate this by selecting a prior for $\alpha$ that keeps it from its boundary. Alternatively, a specialized algorithm will be needed to accommodate the boundary effect.

\section{Simulation Studies}\label{sec:simulations.studies}
In this section we detail three simulation studies that illustrate
different aspects of our modeling approach. In Section \ref{simulation} of
the online supplemental material, we provide additional simulation
results and details regarding a fourth simulation study that
considers  the performance of our method when the response exhibits
spatio-temporal dependence.

\subsection{Simulation 1: Temporal Dependence in Estimated Partitions}

The purpose of the first simulation is to study the accuracy of partition
estimates (i.e., $\hat{\rho}_t$) and how much they change over time. (For
a discussion of how $\hat{\rho}_t$ is obtained, see below.)  In addition,
we explore accuracy in estimating $\mu_{it} = \mu^{\star}_{c_{it}t}$ and
$\alpha_t$. To this end, we considered model \eqref{dat.gen.mech} as a
data generating mechanism to create one hundred datasets with fifty
observations at five time points. We used $tRPM(\bm{\alpha}, M)$  with
$\alpha_t=\alpha$ for all $t$ and $M=1$. We generate synthetic datasets
under  $\alpha \in \{0, 0.1, 0.25, 0.5, 0.75, 0.9, 0.999\}$. For all $i$
and $t$, we set $\sigma_{c_{it}t}^{2\star} = \sigma^2 = 1$, $\tau^2 = 25$,
and $\theta_t = 0$.

To each synthetic data set we fit model \eqref{dat.gen.mech} using the
MCMC algorithm detailed in Section \ref{PostComputation} by collecting
10,000 iterates and discarding the first 5,000 as burn-in and thinning by
5 (resulting in 1,000 MCMC samples). As prior parameters we used $A_{\sigma} = 5$, $A_{\tau} = 10$, $A_{\lambda} = 10$, $m_0 = 0$, $S^2_0 = 100$, $a_{\alpha} = b_{\alpha} = M=1$.  For simplicity we set $\alpha_t = \alpha$ for all $t$.  All partition point
estimates were estimated using the method in the {\tt salso} R
package (\citealt{salso:2020}) with the binder loss function
(\citealt{binder:1978}). To measure similarity between
partitions, we employed the adjusted Rand index
(\citealt{rand:1971,hubert&arabie:1985}) and we used WAIC
(\citealt{gelman&hwang&vehtari:2014}) to measure model fit.

Table \ref{sampleARI} displays the lagged 1 and  4 adjusted Rand index
(ARI) as a function of $\alpha$.  As expected, for both lags, the ARI
increases as $\alpha$ increases.  Also as expected, lagged 4 ARI increases
less as a function of $\alpha$ compared to the lagged 1 ARI.  Note that on
average the lagged 1 ARI for $\alpha \in \{ 0.1, 0.25\}$ is smaller than
that for $\alpha = 0$. This is because the variability associated with
lagged 1 ARI when $\alpha=0$ is much larger than when $\alpha > 0$,
producing a few lagged ARI values that are large. The median of the lagged
ARI values  increase as a function of $\alpha$ monotonically.

\spacingset{1}
\begin{table}[t]
\caption{Adjusted Rand index when comparing $\hat{\rho}_1$ to
$\hat{\rho}_2$ and $\hat{\rho}_1$ to $\hat{\rho}_5$. Note that $ARI(\cdot,
\cdot)$ denotes the adjusted Rand index as a function of two partitions.
Coverage rates for $\alpha$ and $\mu_{it}$ and model fit metrics for
$tRPM(\alpha, M)$ and $CRP(M)$. These values are averaged over the 100
generated data sets. The values in parenthesis are Monte Carlo standard
errors. Note that smaller values of WAIC indicate better fit.}
\vspace{0.25 cm} \centering
\begin{tabular}{l cc cc cc}
  \toprule
  \multicolumn{3}{c}{} &   \multicolumn{2}{c}{Coverage} &   \multicolumn{2}{c}{WAIC} \\ \cmidrule(r){4-5} \cmidrule(r){6-7}
 & $ARI(\hat{\rho}_1, \hat{\rho}_2)$ & $ARI(\hat{\rho}_1, \hat{\rho}_5)$ & $\alpha$ & $\mu_{i t}$ & tRPM & CRP\\
  \midrule
$\alpha=0.0$ 		& 0.05 (0.02) & 0.01 (0.01) & 0.00 (0.00) & 0.93 (0.01) & 896 & 892 \\ 
$\alpha=0.1$ 		& 0.04 (0.01) & 0.00 (0.01) & 0.96 (0.02) & 0.92 (0.01) & 910 & 907 \\ 
$\alpha=0.25$ 		& 0.19 (0.03) & 0.02 (0.02) & 0.90 (0.03) & 0.91 (0.01) & 890 & 891 \\ 
$\alpha=0.5$ 		& 0.44 (0.02) & 0.04 (0.01) & 0.82 (0.04) & 0.91 (0.01) & 881 & 903 \\ 
$\alpha=0.75$ 		& 0.67 (0.02) & 0.23 (0.02) & 0.89 (0.03) & 0.92 (0.01) & 822 & 893 \\ 
$\alpha=0.9$ 		& 0.87 (0.01) & 0.55 (0.02) & 0.90 (0.03) & 0.90 (0.01) & 816 & 896 \\ 
$\alpha=0.9999$ 	& 0.97 (0.01) & 0.93 (0.01) & 0.58 (0.05) & 0.90 (0.01) & 795 & 888 \\ 
   \bottomrule
\end{tabular}
\label{sampleARI}
\end{table}%
\spacingset{1.5}

To study the ability to recover $\mu_{it}$ and $\alpha$,  95\% credible
intervals for each were computed and coverage was estimated.  Results are
provided in Table \ref{sampleARI}. Notice that coverage for $\alpha$ is
low when the true $\alpha$ is at or near the boundary (e.g., $\alpha \in
\{0, 0.9999\})$ which is to be expected. The coverage associated with
$\mu_{it}$ is close to the nominal rate regardless of the value of
$\alpha$.  Therefore, temporal dependence in the partition model does not
adversely impact the ability to estimate individual means.

Lastly, to compare model fit when using $tRPM(\alpha, M)$ as the RPM in
model \eqref{dat.gen.mech} relative to $\rho_t \stackrel{iid}{\sim}
CRP(M)$, we calculated the WAIC for each data set when fitting model
\eqref{dat.gen.mech} under both RPMs. Results are provided in Table
\ref{sampleARI} where each entry is an average WAIC value over all 100
datasets. Notice that, when the independent partitions were used to
generate data (i.e., $\alpha = 0$), modeling partitions independently
produces slightly better model fit as would be expected. But even if
relatively weak temporal dependence exists among the sequence of
partitions, there are gains in modeling the sequence of partitions with
$tRPM(\alpha, M)$, with gains becoming substantial as $\alpha$ increases.

The upshot from this  simulation study is that lagged partition estimates
when employing  $tRPM(\alpha, M)$ display intuitive behavior in that
similarity between partition estimates decreases as lag increases.  In
addition, employing the $tRPM(\alpha, M)$ partition model does not
negatively impact parameter estimation and produces improved model fits
when dependence is present in the sequence of partitions and a minimal
cost in model fit when it is not.

\subsection{Simulation 2: Induced Correlation at the Response Level}
A potential benefit of developing a joint model for partitions is the
ability to accommodate temporal dependence that may exist between $Y_{it}$
and $Y_{it+1}$. To study this, we conducted a small Monte Carlo simulation
study that is comprised of sampling repeatedly from the $tRPM(\alpha, M)$
using the computational approach of Section \ref{PostComputation}. Once
the partition is generated, the temporal dependence among the $\bm{Y}_{i}$
depends on specific model choices for $\mu_{jt}^{\star}$. Here we use
$\mu^{\star}_{jt} \sim N(\phi_1\mu_{jt-1}^{\star}, \tau^2(1-\phi_1^2))$
for $t > 2$,  $j=1, \ldots, k_t$, and $|\phi_1| \le 1$.  For $t=1$ we use
$\mu^{\star}_{j1} \sim  N(0, \tau^2)$ and if $k_{t+1} > k_t$ new
$\mu^{\star}_{jt+1}$ values are drawn from $N(0, \tau^2)$. Now setting
$m=25$, $T=10$,  $\tau=10$, and $\sigma=1$,  100 datasets were generated
for $\phi_1 \in \{0,  0.25, 0.5, 075, 0.9, 1\}$. For each data set
generated,  the lagged auto-correlations among $\bm{Y}_i$ were computed
for $i=1,\ldots, m$. The results found in Figure
\ref{MarginalCorrelations} are the lagged auto-correlations averaged over
the $m$ units for $\alpha \in \{0,  0.25, 0.5, 0.75, 0.9\}$.

As can be seen in Figure \ref{MarginalCorrelations}, when partitions are
independent (i.e., $\alpha = 0$), no correlation propagates to the data
level. The same can be said if atoms are $iid$ (i.e., $\phi_1 = 0$).   As
the temporal dependence among $\mu_{jt}^{\star}$ increases (i.e., $\phi_1$
increases), there is stronger temporal dependence among $Y_{i1}, \ldots,
Y_{iT}$. Notice further that this  dependence persists  longer in time as
$\alpha$ increases as one would expect.

\spacingset{1}
\begin{figure}[t]
\begin{center}
\includegraphics[scale=0.5]{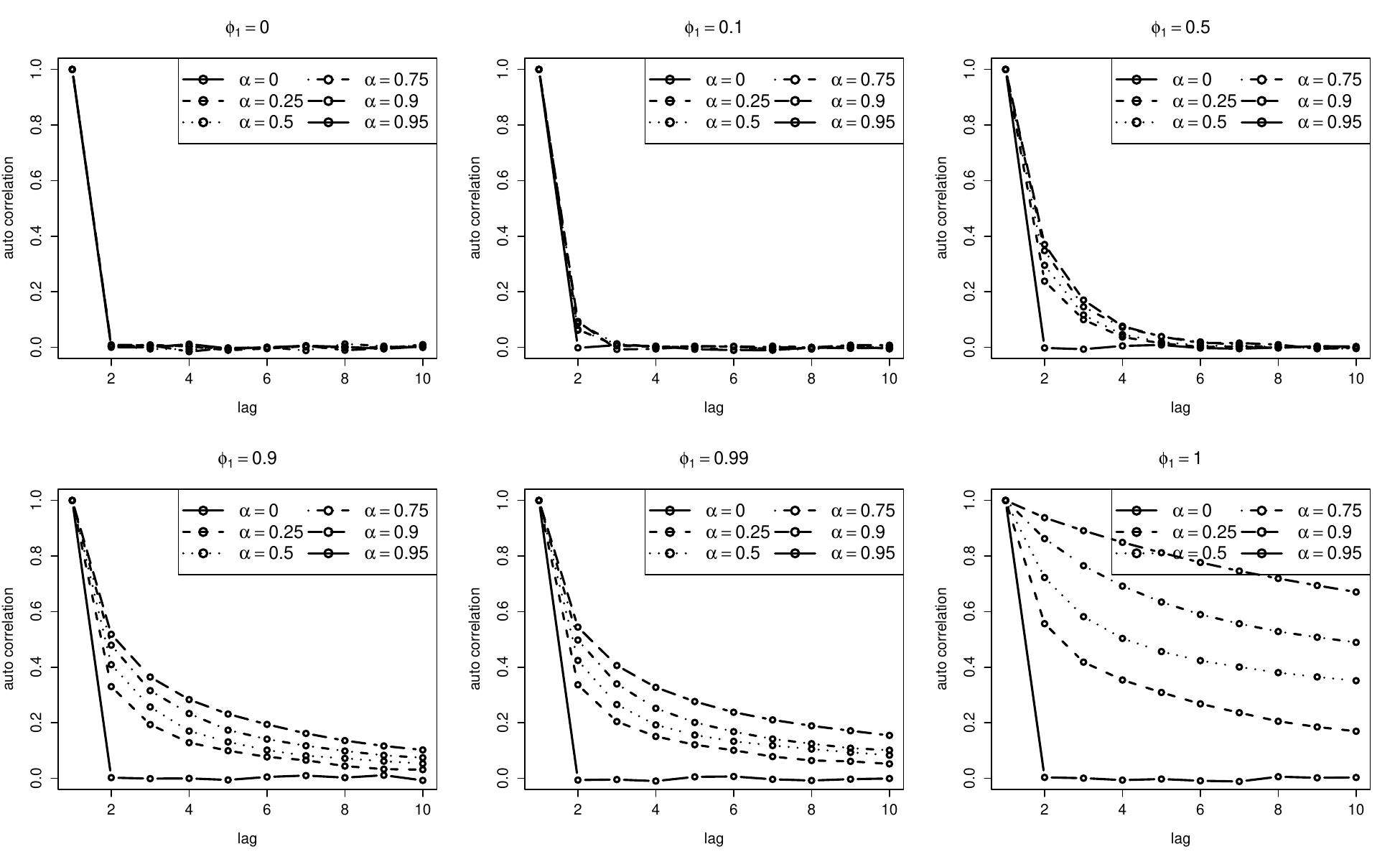}
\caption{Lagged auto-correlations among the $(Y_{i1}, \ldots, Y_{iT})$ when
modeling $\mu^{\star}_{jt}$ with an AR(1) type structure. }
\label{MarginalCorrelations}
\end{center}
\end{figure}
\spacingset{1.5}

\subsection{Simulation 3: AR(1)-type synthetic data} \label{sim.study3}
In our final simulation experiment, we consider data generated from an
AR(1) process. To create synthetic datasets for the $i$th unit, we employ
the following as a data generating mechanism
\begin{align*}
Y_{it} =  \mu_{c_{it}} + \omega Y_{it-1} + \epsilon_{it}, \ \mbox{for $i=1, \ldots, m$, and $t = 1, \ldots, T$},
\end{align*}
where $|\omega| < 1$ and $\epsilon_{it} \sim N(0, v^2)$. We consider
synthetic datasets with four clusters so that $c_{it} \in \{1,2,3,4\}$
corresponding to $\mu\in\{-2, 0, 2, 4\}$.  The four clusters are formed by
dividing the $m=100$ units into equal groups of 25.  At time points $t=2,
\ldots, T$, four units from each cluster are shifted to other clusters in
a systematic way so that clusters change over time.    An example of the
type of data this procedures creates can be seen in Figure
\ref{synthetic_data1} of the supplementary material.   Data are generated
using the function {\tt arima.sim} in {\tt R} (\citealt{R:2020}) under
$\omega \in \{0, 0.1, 0.25, 0.5, 0.75, 0.9\}$, $v^2 \in \{0.5^2, 1^2\}$,
and $T \in \{5, 10\}$.   A total of 100 datasets are created under each
scenario (totaling 28) and to each we fit the following competing models.

\begin{enumerate}
\item  A weighted dependent Dirichlet process (wddp) described in
    \cite{quintana2020dependent} and chapter 4.4.4 of
    \cite{BNPbook:2015}.  This model incorporates time in the weights
    of a DDP.  As such, we fit this model to a concatenated version of
    the data  ($Y_i, t_i)$ for $i = 1, \ldots, Tm$.  Specific details
    of this procedure are provided in Section \ref{simstudy3_cont} of
    the online supplementary material.
\item A linear dependent Dirichlet process (lddp) described in
    \cite{quintana2020dependent} and chapter 4.4.2 of
    \cite{BNPbook:2015}.  This model incorporates time in the atoms of
    a DDP.  As in the wddp, to this model we concatenated version of
    the data  ($Y_i, t_i)$ for $i = 1, \ldots, Tm$.  Specific details
    of this procedure are also provided in Section
    \ref{simstudy3_cont} of the online supplementary material.
\item A Griffiths-Milne dependent Dirichlet process (gmddp) mixture.
    This model was fit using the {\tt BNPmix} package in {\tt R}
    \citep{BNPmix}.  See \cite{BNPmix_JSS} for specific model details.
\item A temporally independent $CRP(M)$ model (ind\_crp). This
    procedure corresponds to $\alpha = 0$ and is a special case of our approach and that proposed in  \citet{Caron:2007}. The exact details of this
    model are also provided in  Section \ref{simstudy3_cont} of the
    online supplementary material. 
\item A temporally static $CRP(M)$ model (static\_crp). This procedure
    corresponds to $\alpha= 1$ and is also a special case of the model
    detailed in \citet{Caron:2007}. Like the wddp and lddp, this
    procedure is fit to  a concatenated version of the data  ($Y_i$,
     for $i = 1, \ldots, Tm$.  See  Section \ref{simstudy3_cont}
    of the online supplementary material for more details.
\end{enumerate}

Results from this simulation study are presented in Figure
\ref{fig:sim.study3}. The left plot in the figure displays the WAIC model
fit metric for each procedure averaged over the 100 synthetic data sets,
while the right plot displays the ARI value. The ARI values were produced
by calculating ARI for each MCMC sample from the posterior distribution of
$\rho_t$. This was done separately for each $t = 1, \ldots, T$ and then
averaged across time and MCMC samples. From the left plot, our method (crpm)
is superior to all other methods in terms of the model fit metric WAIC,
save for the lddp method. Against the lddp method, the drpm produces
better fits when the data are generated with  high auto-correlation and
larger data noise.   This is to be expected as the drpm only incorporates
temporal information in the prior on partitions while the lddp includes
this information in the atoms (i.e., likelihood). That said, in terms of
partition estimation, the drpm easily outperforms all other competitors in
terms of ARI. The fact that drpm outperformed wddp and lddp in terms of
ARI is not surprising as the later methods incorporate time differently
compared to the drpm.   However, the drpm and gmddp methods both treat
time similarly, but the former from a partition perspective while the
latter from a random probability measure perspective. This highlights the
fact that when partitions are of interest, modeling them directly provides
benefit.

\spacingset{1}
\begin{figure}
\centering
\begin{subfigure}{.5\textwidth}
  \centering
  \includegraphics[scale=0.5]{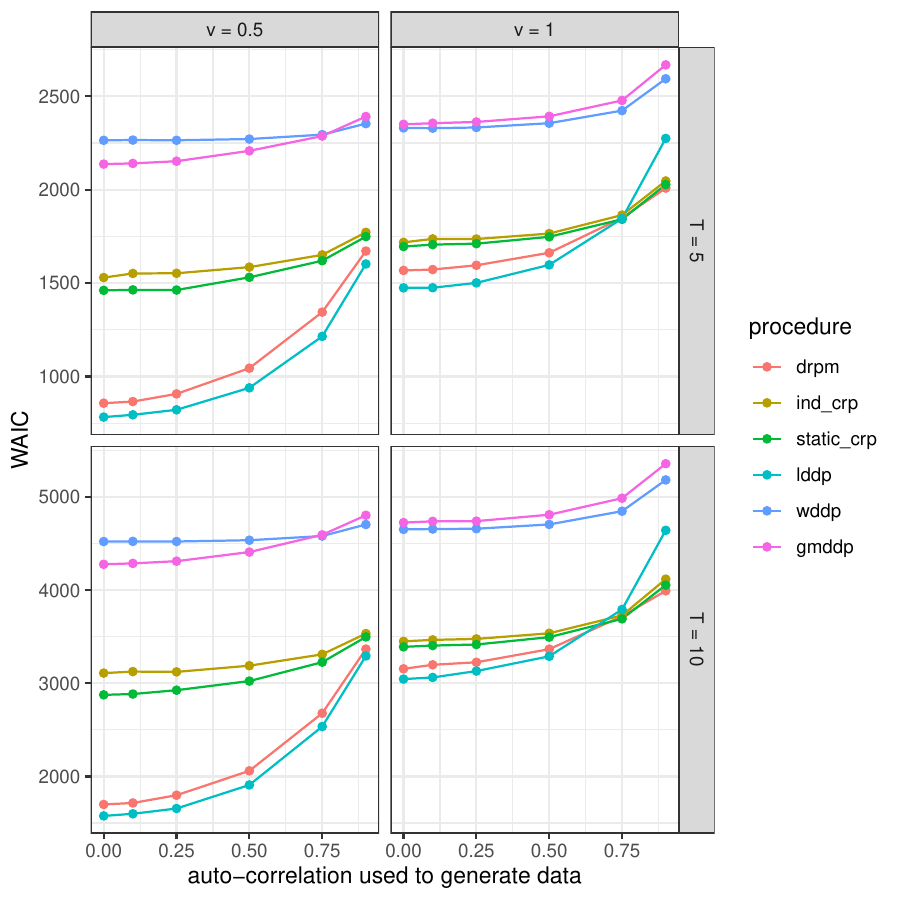}
  \label{fig:sub1}
\end{subfigure}%
\begin{subfigure}{.5\textwidth}
  \centering
  \includegraphics[scale=0.5]{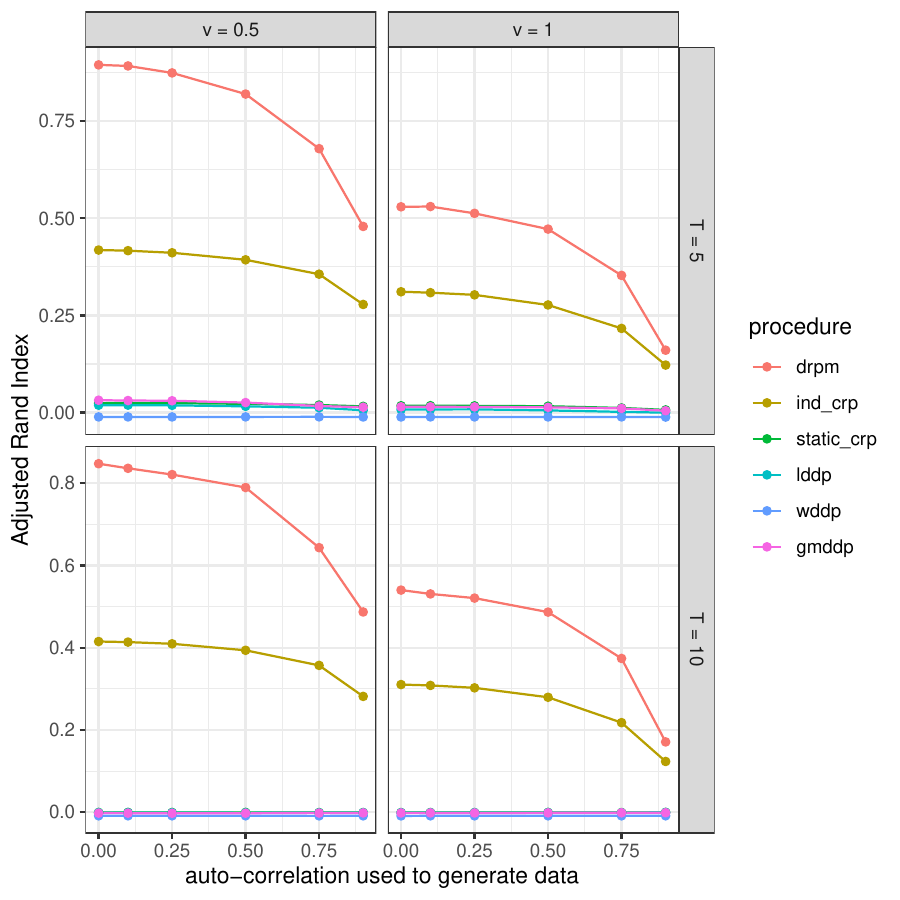}
  \label{fig:sub2}
\end{subfigure}
\caption{Results from simulation study with data that contains AR(1)-type temporal
correlation.  The left plot corresponds to results for the model fit
metric WAIC and the right plot displays results associated with ARI.}
\label{fig:sim.study3}
\end{figure}
\spacingset{1.5}

\section{Application} \label{applications}
In this section we apply our method to a real-world data set coming from
the field of environmental science. A second application in education is
provided in Section \ref{SIMCE.application} of the online supplementary
material. As mentioned previously, once a partition model is specified,
there is quite a bit of flexibility regarding how (or if) temporal
dependence is incorporated in other parts of a hierarchical model.  To
illustrate this, we incorporate temporal dependence in three places of the
hierarchical model we construct.

As part of preliminary exploratory data analysis (not shown), we examined
serial dependence for each experimental unit (monitoring station), and
concluded that they all exhibited a particular type of temporal
dependence. Because of this,  we introduce a unit-specific temporal
dependence parameter $|\eta_{1i}| \le 1$ and model observations from a
single unit over time ($Y_{1i}, \ldots, Y_{iT}$) with an AR(1) structure.
In addition, motivated by a desire for parsimony, we employed a Laplace
prior for  $\eta_{1i}$. Finally, to permit the temporal dependence in the
partition model to propagate through the hierarchical model,  we assume
an AR(1) structure for the $\theta_t$'s.  The full  hierarchical model
is detailed in \eqref{FullModel}.
\begin{equation}\label{FullModel}
\begin{aligned}
Y_{it} \mid   Y_{i t-1}, \bm{\mu}^{\star}_{t},\bm{\sigma}^{2\star}_t, \bm{\eta}, \bm{c}_{t} & \stackrel{ind}{\sim} N(\mu^{\star}_{c_{it}t} + \eta_{1i}Y_{i t-1},\sigma_{c_{it}t}^{2\star}(1-\eta_{1i}^2)), \\
Y_{i1} & \stackrel{ind}{\sim} N(\mu^{\star}_{c_{i1}1}, \sigma_{c_{i1}1}^{2\star}),\\
\xi_i = \mbox{Logit}(0.5(\eta_{1i} + 1)) & \stackrel{iid}{\sim} Laplace(a,b), \\
(\mu_{jt}^{\star}, \sigma^{\star}_{jt}) & \stackrel{ind}{\sim} N(\theta_t, \tau_t^2) \times UN(0,A_{\sigma}),  \\
\theta_t \mid  \theta_{t-1} & \stackrel{ind}{\sim} N((1-\phi_1)\phi_0 + \phi_1\theta_{t-1}, \lambda^2(1-\phi_1^2)), \\
(\theta_1, \tau_t) & \sim N(\phi_0, \lambda^2) \times UN(0,A_{\tau}),\\
(\phi_0, \phi_1, \lambda) & \sim N(m_0, s_0^2) \times UN(-1,1) \times UN(0,A_{\lambda}), \\
\{\bm{c}_{t}, \ldots, \bm{c}_{T}\} & \sim tRPM(\bm{\alpha}, M), \ \mbox{with $\alpha_t \stackrel{iid}{\sim} Beta(a_{\alpha}, b_{\alpha})$},
\end{aligned}
\end{equation}
where all Roman letters correspond to parameters that are user supplied.
There are a number of special cases embedded in our
hierarchical model. For example, $\eta_{i1} = 0$ for all $i$ results in
conditionally independent observations. Further, $\phi_1 = 0$ results in
independent atoms and $\alpha_t = 0$ for all $t$ in independent partitions
over time. The model \eqref{dat.gen.mech} used in the simulation
studies is a special case of \eqref{FullModel},  obtained by setting
$\phi_1 = 0$ and $\eta_{i1} = 0$ for all $i$.

\subsection{Rural Background PM$_{10}$ Data Application} \label{pm10.application}
The rural background PM$_{10}$ data is taken  from the European air
quality database. These data are comprised of the daily measurements of
particulate matter with a diameter less than 10 $\mu$m from rural
background stations in Germany and are publicly available in the {\tt
gstat} package  (\citealt{gstat:2016}) found on CRAN in {\tt R}
(\citealt{R:2020}). We focus on average monthly  PM$_{10}$ measures from
the year 2005 ( i.e., $T=12$). Of the 69 stations, 9 were removed because
of missing values.

We fit the hierarchical model \eqref{FullModel} to these data and consider
all the possible special cases (i.e., $\eta_{1i} = 0$ or not, $\phi_1 = 0$
or not, $\alpha_t = 0$ or not).  This resulted in 8 total models that were
fit by collecting 1,000 MCMC iterates after discarding the first 10,000 as
burn-in and thinning by 40 (i.e., 50,000 total MCMC samples were collected).  Running the algorithm for 50,000 samples on a laptop with 16GB of RAM took between 1 and 2.5 minutes.  We use the following prior values: $A_{\sigma} = 10$, $A_{\tau} = A_{\lambda} = 5$, $m_0 = 0$, $s^2 = 100$, $a = 0$, $b=1$, and $a_{\alpha} = b_{\alpha} =
2$.  The prior for $\alpha_t$ was specified to encourage $\alpha$ from approaching 1.   The WAIC and log pseudo marginal likelihood (LPML) for each model are
presented in Table \ref{PM10}.  To improve the computational stability of the LPML, for each model fit, the MCMC iterates that correspond to 0.5\% of the smallest likelihood values were not included in the calculation of LPML.  

First we note that among all the model fits, employing a variant of
$tRPM(\bm{\alpha}, M)$ (i.e., rows with ``Yes'' in the ``In Partition''
column)  improves model fit. The best performing model in terms
of WAIC includes temporal dependence  in the partition model only, while that for LPML includes temporal dependence in the partition model
and in the likelihood.

Now focusing on partition inference,  we provide Figure \ref{LaggedARI}.
This figure displays the lagged ARI values for each of the 8 models.
Notice that when partitions are modeled independently (first or third rows
of Figure \ref{LaggedARI}) then partitions evolve over time quite
erratically in the sense that the cluster configuration can change
dramatically from one time point to the next. However, when employing
$tRPM(\bm{\alpha}, M)$ (second row of Figure \ref{LaggedARI}) the
partitions seemed to evolve much more ``smoothly'' as there is less
drastic changes in cluster configuration. In fact the model that produces
the best model fit metrics (right most plot of second row) seems to
produce partitions that change quite gently over time as desired.

\spacingset{1}
\begin{table}
\caption{PM$_{10}$ data: Results of model fitting. The bold font
identifies best model fits in terms of LPML and WAIC. Higher values for
LPML indicate better fit while lower values for WAIC indicate better fit.}
\centering
\begin{tabular}{ccc  cc  }
\multicolumn{3}{c}{Temporal Dependence In} &   \\ \cmidrule(r){1-3}
Likelihood ($\eta_{1i}$) & Atoms ($\phi_{1}$) & Partition ($\alpha_t$) 	& LPML & WAIC  \\ \midrule
No & No & No  						& -1814 & 3683  \\
No & No & Yes  					& -1656 & \bf{3031}  \\ \midrule
No & Yes & No  					& -1752 & 3539  \\
No & Yes & Yes  					& -1644 & 3271  \\   \midrule
Yes & No & No  					& -1704 & 3554  \\
Yes & No & Yes  					& \bf{-1578} & 3186  \\  \midrule
Yes & Yes & No  					& -1706 & 3544  \\
Yes & Yes & Yes  					& -1586 & 3153  \\   \bottomrule
\end{tabular}
\label{PM10}
\end{table}%
\spacingset{1.5}


\spacingset{1}
\begin{figure}
\centering
\includegraphics[scale=0.45]{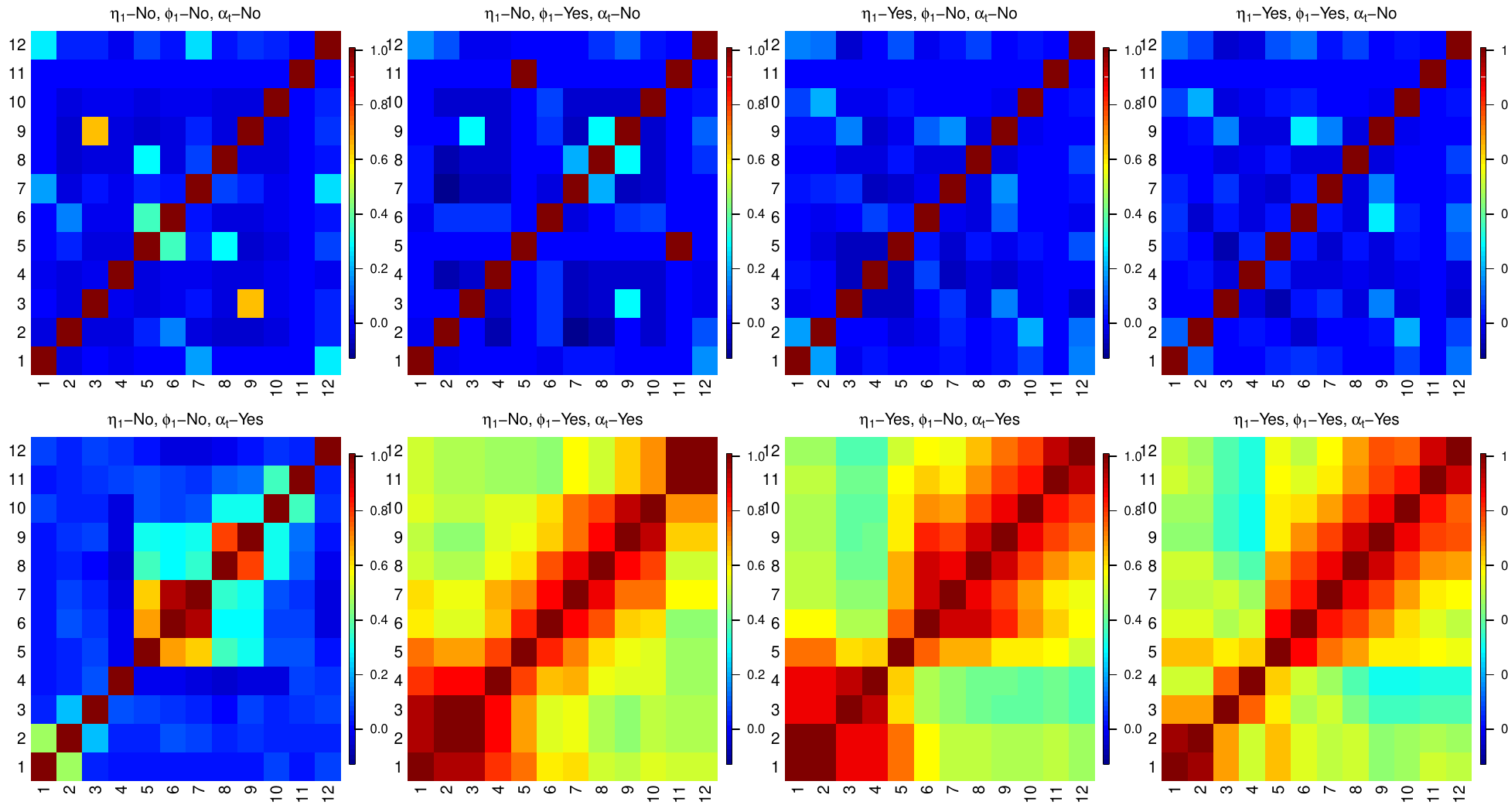}
\caption{PM$_{10}$ data. Each figure is a summary of the lagged $ARI$ values
corresponding to the 8 models in Table \ref{PM10}. At
each time point the partition was estimated using the {\tt salso} function
in the {\tt salso} R package (\citealt{salso:2020}) based on binder loss.}
\label{LaggedARI}
\end{figure}
\spacingset{1.5}

Finally, we provide Figure \ref{Partitions0verTime} as a means to
visualize how estimated partitions based on our joint partition model
evolve over time relative to modeling partitions with an $iid$ model. Each
plot in Figure \ref{Partitions0verTime} displays the estimated partition
at each time point. Each color represents a cluster and each number
corresponds to a specific measuring station. The plots illustrate the
sequential nature of cluster forming with the first cluster always
containing the first measuring station, the next cluster is formed by the
first station not included in the first cluster and so on.  The plots in
the right column correspond to using $tRPM(\bm{\alpha}, M)$ to jointly
model partitions while those on the right employ $\rho_t
\stackrel{iid}{\sim} CRP(M)$. It is evident from Figure
\ref{Partitions0verTime} that from one time point to the next that
partitions based on our construction evolve much more gently over time.
This more closely mimics how PM$_{10}$ measurements would evolve as a
function of time compared to the {\it iid} model.

\spacingset{1}
\begin{figure}
\centering
\includegraphics[scale=0.45]{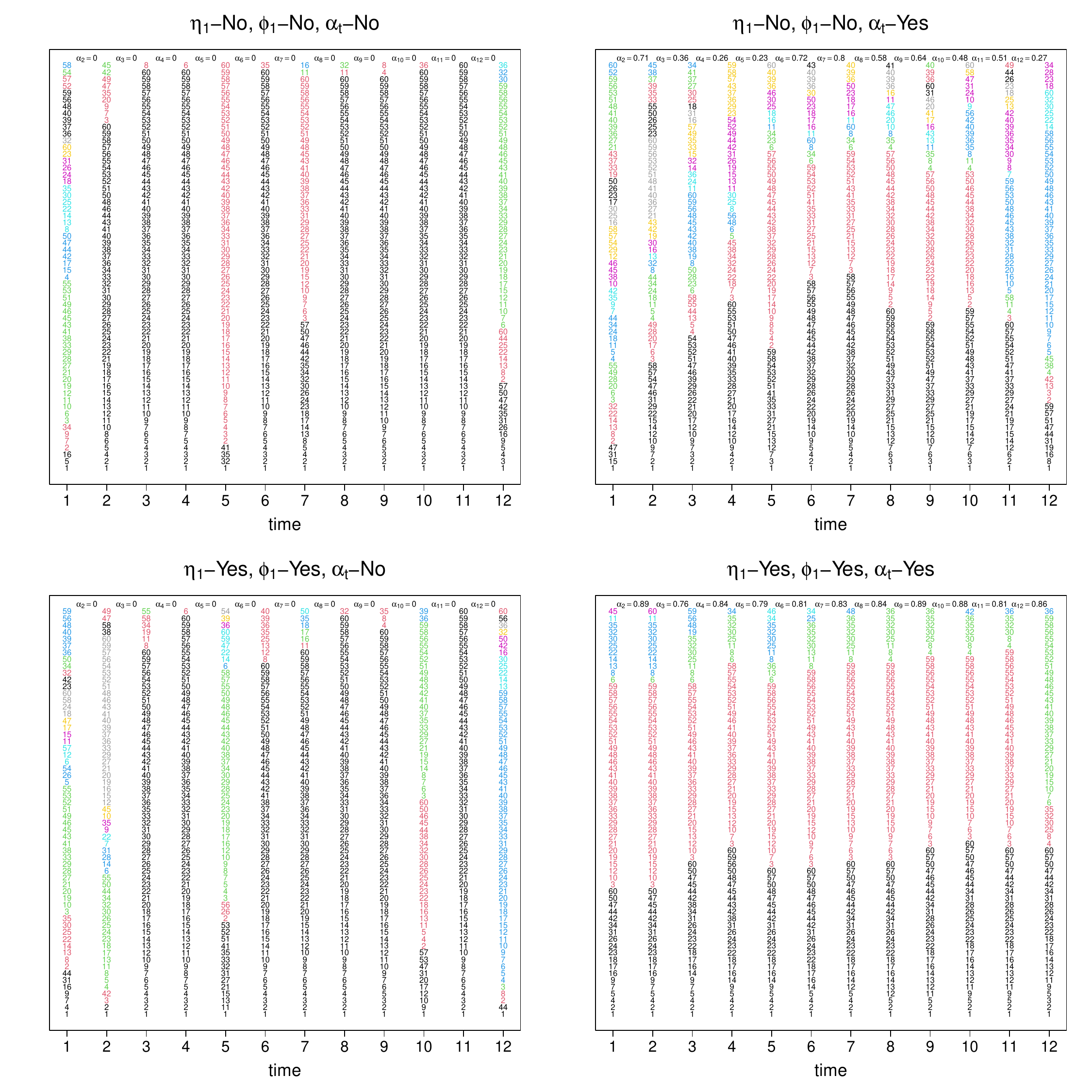}
\caption{Each plot displays the estimated  partition at each time point.  Plots
in the left column are based on $\rho_t \stackrel{iid}{\sim} CRP(M)$ while
those in the right column are based on  $(\rho_1, \ldots, \rho_T) \sim
tRPM(\bm{\alpha}, M)$.  Clusters are highlighted by color and each
number corresponds to a specific monitoring station. }
\label{Partitions0verTime}
\end{figure}
\spacingset{1.5}

\subsection{Extensions to the Joint Partition Model}

Based on our generic joint model construction, it is straightforward to
incorporate other information in the partition model such as covariates. For example, in the
data application of \ref{pm10.application} each monitoring station's
spatial coordinates were recorded. Incorporating spatial dependence in our
analysis of the PM$_{10}$ data can be easily accommodated via the EPPF in
our construction. This would result in spatially informed clusters that
evolve over time. If the spatially referenced EPPF preserves sample size
consistency, then Proposition \ref{proposition1} still holds.  One such
EPPF is part of the spatial product partition model (sPPM) class developed
in \cite{page&quintana:2016}.   To illustrate the ease of incorporating
space in our model construction, here we model the PM$_{10}$ data using
model \eqref{FullModel} but rather than use the $tRPM(\bm{\alpha}, M)$ to
model the sequence of partitions, we use a version of our dependent
partition model that employs the sPPM.

In order to introduce the sPPM, let $\bm{s}_i$ denote the spatial
coordinates of the $i$th item (note that these coordinates do not change
over time) and let $\bm{s}^{\star}_{jt}$ be the subset of spatial
coordinates that belong to the $j$th cluster at time $t$. Then we express
the EPPF of the $t$th  partition with the following product form
\begin{align} \label{sPPM}
\text{Pr}(\rho_t \mid  \nu_0, M) \propto \prod_{j=1}^{k_t} c(S_{jt}\mid M) g(\bm{s}^{\star}_{jt}\mid \nu_0).
\end{align}
Here $c(\cdot\mid M) \ge 0$ is called the cohesion and is a set function
that produces cluster weights {\it a priori}.  We consider the cohesion
$c(S_{jt}\mid M) = M \times (|S_{jt}| - 1)!$ as it has connections with
the CRP making this version of the sPPM a type of spatially re-weighted
CRP. The similarity function $g(\cdot\mid \nu_0)$ is a set function
parametrized by $\nu_0$ that  measures the compactness of the spatial
coordinates in $\bm{s}^{\star}_{jt}$ producing higher values if the
spatial coordinates in $\bm{s}^{\star}_{jt}$ are close to each other. Not
all similarity functions preserve sample size consistency so to ensure
this, after standardizing spatial locations, we employ
\begin{align}\label{similarity.function}
g(\bm{s}^{\star}_{jt}\mid \nu_0) = \int \prod_{i \in S_{jt}} N(\bm{s}_i \mid  \bm{m},\bm{V}) NIW(\bm{m},\bm{V}\mid  \bm{0}, 1, \nu_0, \bm{I})\, d\bm{m}\,d\bm{V},
\end{align}
where $N(\cdot \mid  \bm{m},\bm{V})$ denotes a bivariate normal density
and $NIW(\cdot,\cdot\mid  \bm{0}, 1, \nu_0, \bm{I})$ a
normal-inverse-Wishart density with mean $\bm{0}$, scale equal to 1,
inverse scale matrix equal to $\bm{I}$, and $\nu_0$ being the
user-supplied degrees of freedom.  Note that larger values of $\nu_0$
increase spatial influence on partition probabilities. For more details on
why this formulation preserves sample size consistency, see
\cite{PPMxMullerQuintanaRosner} and \cite{quintana&loschi&page:2018}. For
more information regarding the impact of $\nu_0$ on product form of the
partition model, see \cite{page&quintana:2016, page&quintana:2018}. We
will use $stRPM(\bm{\alpha}, \nu_0, M)$ to denote our spatio-temporal
random partition model \eqref{joint.joint.model} parameterized by
$\alpha_1, \ldots, \alpha_T$ and EPPF detailed in \eqref{sPPM} and
\eqref{similarity.function}.

As in the previous section, we fit  model \eqref{FullModel} to the
PM$_{10}$ data but replacing  $tRPM(\bm{\alpha}, M)$ with the
$stRPM(\bm{\alpha}, \nu_0, M)$.  Also as before, we consider all the
possible special cases of the model (i.e., $\eta_{1i} = 0$ or not, $\phi_1
= 0$ or not, $\alpha_t = 0$ or not).  This resulted in 8 total models that
were fit by collecting 1,000 MCMC iterates after discarding the first
10,000 as burn-in and thinning by 40.  Fitting the eight models based on the $stRPM(\bm{\alpha}, \nu_0, M)$ took between 20 and 57 minutes.  The prior values employed were the
same as before with the addition of setting $\nu_0 = 5$ which places in the partition prior moderate weight on spatial locations.

Incorporating space in the partition model seems to provide benefit in
terms of model fit as the LPML and WAIC values associated with the model
that includes space in the partition model and temporal dependence in all
levels of the model fits best in terms of LPML with -1487 compared to -1586 listed
in Table \ref{PM10} and temporal dependence in partition and atoms fits best with regards to WAIC with 3140 compared to  3271 listed
in Table \ref{PM10}. Additionally, Figure \ref{LaggedARIOnlySpace}
displays the lagged ARI values for the 8 models that include space. As in
Figure \ref{LaggedARI} when there is no temporal dependence in the
partition model, then the estimated partitions exhibit no temporal
dependence. However, when time and space are included then there is clear
dependence between partitions as a function of time. Further, the
dependence between partitions seems to decay faster with space and time
included in the partition model compared to just time.

\spacingset{1}
\begin{figure}
\centering
\includegraphics[scale=0.45]{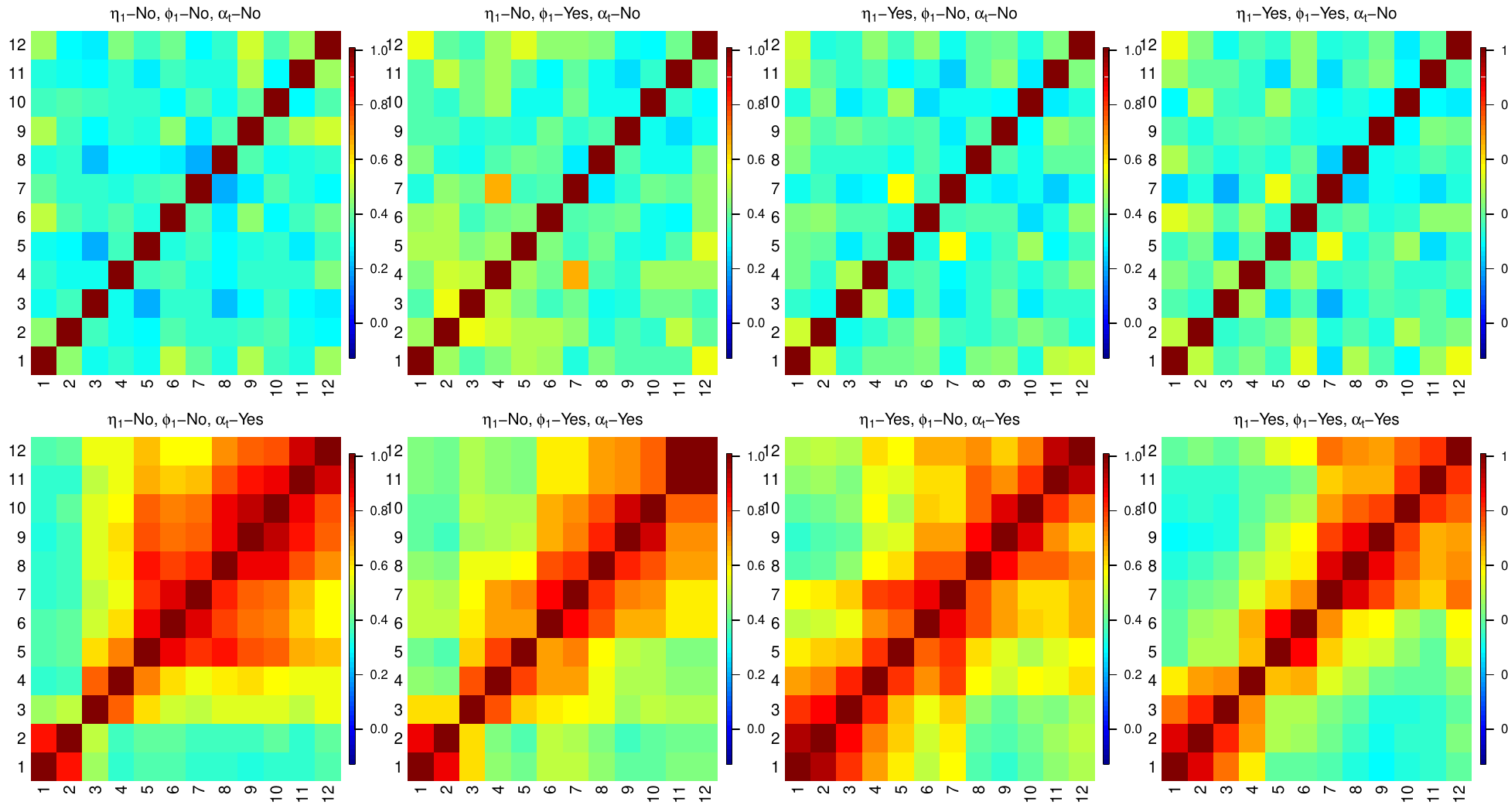}
\caption{PM$_{10}$ data. Each figure is a summary of the lagged $ARI$ values
corresponding to the 8 models that are detailed in Table \ref{PM10} except
that space is also included in the dependent random partition model. At
each time point the partition was estimated using the {\tt salso} function
in the {\tt salso} R package (\citealt{salso:2020}) based on binder loss.}
\label{LaggedARIOnlySpace}
\end{figure}
\spacingset{1.5}

Finally, we provide Figure \ref{ClusterEstimates} which displays the
estimated spatially referenced partitions at each time point based on the
model that achieved the best fit (space in the partition model  and
temporal dependence in all levels of the model).  The size of each point
in the figure is proportional to the PM$_{10}$ measured at the particular
station and each depicts a cluster.  To make connections with Figure
\ref{Partitions0verTime} each monitoring station is labeled with the same
number as before.  Notice that there are clear similarities from one time
point to the next for most months. That said, there are two time periods
for which changes in the PM$_{10}$ are more drastic relative to the
previous time period (e.g., August to September). In these months the
estimated $\alpha_t$ is quite a bit smaller and as a result, the estimated
partitions are more different.

\spacingset{1}
\begin{figure}[htp]
\centering
\includegraphics[scale=0.6]{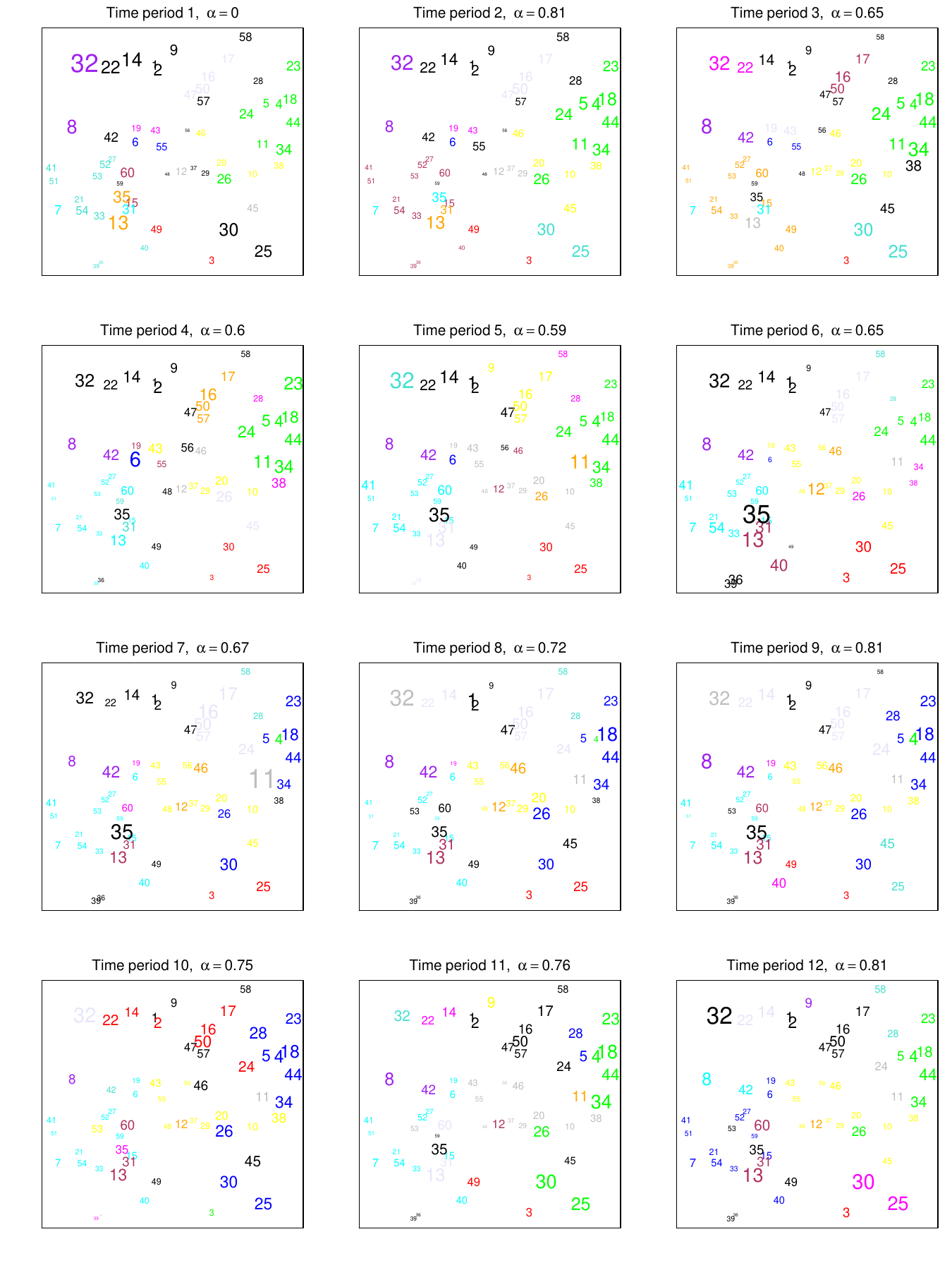}
\caption{PM$_{10}$ data. Graphical display for spatially-referenced estimated
partitions for each time point based on the model that produced best fit
(space in the partition model and temporal dependence in all levels of
the model). The size of a point is proportional to the PM$_{10}$ measured at
that station. Clusters are identified by color. Each monitoring station
is labeled by the same numbers as in Figure \ref{Partitions0verTime}.  At
each time point the partition was estimated using the {\tt salso} function
in the {\tt salso} R package (\citealt{salso:2020}) based on binder loss. }
\label{ClusterEstimates}
\end{figure}
\spacingset{1.5}

\section{Conclusions} \label{conclusions}

We developed a joint probability model for a sequence of partitions that
explicitly considers temporal dependence among the partitions.  Further,
we showed that our methodology is capable of accommodating partitions that
evolve slowly  over time in that the adjusted Rand index between estimated
partitions decays as the lag in time increases. We also showed that if
partitions are indeed independent over time, then employing our joint
partition prior regardless results in a minimal cost in terms of model
fit.

The predictive nature of the temporal prior on a sequence of random
partitions we have presented has a first-order Markovian structure.
Various extensions can be considered, such as adding higher order
dependence across time or dependence in baseline covariates. All of these
cases would build on our constructive definition, as extra refinements of
the basic idea of carrying smooth transitions on time (or time and space).
Lastly, the Markovian structure could be exploited to carry out predictive
inference as well.

\bigskip
\begin{center}
{\large\bf SUPPLEMENTARY MATERIAL}
\end{center}

\begin{description}

\item[Supplementary Material:] Online supplementary material file that
    contains proofs of
    propositions, computation details, additional simulation, and application results.

\item[R-package:] The R-package {\tt drpm} contains the function {\tt
    drpm} that fits all models described in the article.

\end{description}

\bibliographystyle{jasa3}

\bibliography{reference}
\end{document}


\def\spacingset#1{\renewcommand{\baselinestretch}%
{#1}\small\normalsize} \spacingset{1}


\if1\blind
{
  \title{\bf Supplementary Material:  \\ Dependent Modeling of Temporal Sequences of Random Partitions}
  \author{Garritt L. Page\\
    Brigham Young University, Provo, USA\\
    BCAM - Basque Center of Applied Mathematics, Bilbao, Spain\\
    and \\
    Fernando A. Quintana \thanks{ Partially supported by grant FONDECYT 1180034 and by Iniciativa Cient\'{\i}fica Milenio - Minecon N\'ucleo Milenio MiDaS}\hspace{.2cm}\\
    Pontificia Universidad Cat\'{o}lica de Chile, Santiago, Chile \\
    and \\
    David B. Dahl \\
    Brigham Young University, Provo, USA.}
  \maketitle
} \fi

\if0\blind
{
  \bigskip
  \bigskip
  \bigskip
  \begin{center}
    {\LARGE\bf Dependent Sequence of  Random Partition Models}
\end{center}
  \medskip
} \fi

\doublespacing

This document contains supplementary material to the paper ``Dependent Modeling of Temporal Sequences of Random Partitions''.

\section{Proofs of Propositions}

In this section of the supplementary material we provide proofs of the two
propositions described in the main article

\subsection{Proof of Proposition 1}
\begin{proof}
For clarity, here we introduce notation that highlights the dependence of
partitions on sample size.  For example, $\rho_{t,m} =
(S_{1,t},\ldots,S_{k_t(m), t })$ and $[m]=\{1,\ldots,m\}$. By assumption
$\text{Pr}(\rho_{1,m})$ is specified by means of an EPPF which we now
construct. Denote $\mathbb{N}^*=\cup_{k=0}^{\infty}\mathbb{N}^k$, and
identify any $\bm{n}=(n_1,\ldots,n_k)\in\mathbb{N}^*$ with  the infinite
sequence $(n_1,\ldots,n_k,0,0,\ldots)$. Given $\bm{n}\in\mathbb{N}^*$, let
$k(\bm{n})$ denote the number of non-zero entries in $\bm{n}$ and denote
by $\bm{n}^{j+}$ the result of incrementing $\bm{n}$'s $j$th component
(i.e., $n_j$)  by 1, with $1\le j\le k(\bm{n})+1$. An EPPF is then any
function $r:\mathbb{N}^* \longrightarrow [0,1]$ that is symmetric in its
arguments and where
\begin{equation}\label{eq:EPPF}
r(1)=1\qquad\mbox{and}\qquad r(\bm{n})=\sum_{j=1}^{k(\bm{n})+1}r(\bm{n}^{j+})
\qquad \mbox{for all $\bm{n}\in\mathbb{N}^*$.}
\end{equation}
Condition \eqref{eq:EPPF} implies that a EPPF is sample size consistent,
i.e., marginalizing the $(n+1)$st element leads to the model for $n$
elements. The EPPF also implies exchangeability of configurations in the
sense that a EPPF is invariant under permutations of the elements that
keep the cluster sizes unaltered. We also note that any valid EPPF defines
a predictive rule of the form
\begin{equation}\label{eq:PPF}
r_j(\bm{n})=\frac{r(\bm{n}^{j+})}{r(\bm{n})},\qquad\mbox{for $1\le j\le k(\bm{n})+1$,}
\end{equation}
where it is assumed that $r(\bm{n})>0$ and $r_j(\bm{n})$  represents the
probability of a new element joining the $j$th already existing cluster,
for $1\le j\le k(\bm{n})$, or starting a new one (the $k(\bm{n})+1$). The
one-step rule \eqref{eq:PPF} can also be extended to predictions of two or
more elements by simply iterating the one-step rule as many times as
needed.  Now, given an EPPF $r$, we have that
\begin{equation}\label{eq:priorrho1}
\text{Pr}(\rho_{1,m}=(S_{1,1},\ldots,S_{k_1(m),1}))=r(n_{1,1},\ldots,n_{k_1(m),1}).
\end{equation}

To prove the result, it suffices to show that it holds for $\rho_{2,m}$
and then by induction the result holds generally. Denote by
$[\Gamma]=\{i\in\{1,\ldots,m\}: \,\gamma_{i2}=0\}$ the (random) set of
elements removed from $\rho_{1,m}$. Then, $\rho_{1,m}^{-N_{02}}$ is a
partition of the elements of $[m]-[\Gamma]$ (where as before $N_{02} =
\sum_{j = 1}^m I[\gamma_{j2} = 0]$). By exchangeability and the fact that
an EPPF is sample size consistent, we have that for any partition
$S^-_1,\ldots,S^-_{k([m]-[\Gamma])}$ of $[m]-[\Gamma]$:
\begin{align*}
\text{Pr}(\rho_{2,m}^{-N_{02}}=(S^-_1,\ldots,S^-_{k([m]-[\Gamma])})\mid [\Gamma]) & =  \text{Pr}(\rho_{1,m}^{-N_{02}}=(S^-_1,\ldots,S^-_{k([m]-[\Gamma])})\mid [\Gamma])\\
& =r(|S^-_1|,\ldots,|S^-_{k([m]-[\Gamma])}|),
\end{align*}
where $|S_j|$ is the number of elements in $S_j$. In addition, and again
by exchangeability and sample size consistency, the predictive rule
starting from $[m]-[\Gamma]$ (or from any subset of $[m]$ for that matter)
depends only on the sizes of the subsets in that partition. Thus,
conditioning on all reallocation configurations and initial partition
after subject removal we have:
\begin{align*}
\begin{split}
\text{Pr}(\rho_{2,m}=(S_1,\ldots,S_k))&= \sum_{[\Gamma]} \sum_{\rho_{2,m}^{-N_{02}}}   \text{Pr}(\rho_{2,m}=(S_1,\ldots,S_k)\mid[\Gamma],\,\rho_{2,m}^{-N_{02}}) \times \\
  & \qquad \qquad \qquad \text{Pr}(\rho_{2,m}^{-N_{02}}\mid [\Gamma])\text{Pr}([\Gamma]),
\end{split} \\
\begin{split}
& =\sum_{[\Gamma]} \sum_{\rho_{1,m}^{-N_{02}}}  \text{Pr}(\rho_{1,m}=(S_1,\ldots,S_k)\mid[\Gamma],\,\rho_{1,m}^{-N_{02}})  \times \\
  & \qquad \qquad \qquad \text{Pr}(\rho_{1,m}^{-N_{02}}\mid [\Gamma])  \text{Pr}([\Gamma]),
\end{split} \\
& =\text{Pr}(\rho_{1,m}=(S_1,\ldots,S_k)),
\end{align*}
where the second to last equality follows from the constructive
description given earlier and the properties of the EPPF. The result then
follows.
\end{proof}

\subsection{Proof of Proposition 2}
\begin{proof}
The proof proceeds by direct calculations.
\begin{itemize}
\item[(a)] Let $\gamma_i=1$ if unit $i\in[m]$ is not relocated. Note that $\gamma_1,\gamma_2\stackrel{iid}{\sim}
Ber(\alpha)$. By definition, $P(c_{11}=c_{21})=\frac{1}{M+1}$ and $P(c_{11}\ne c_{21})=\frac{M}{M+1}$.
By conditioning on $\gamma_1,\gamma_2$ and $c_{11},c_{21}$ we get
$$P\left(c_{12}=c_{22} \mid c_{11}=c_{21},\left(\gamma_{1}, \gamma_{2}\right)\right)=\left\{
\begin{array}{cl}
1 & \mbox{if $\left(\gamma_{1}, \gamma_{2}\right)=(1,1)$} \\
\frac{1}{M+1} & \mbox{otherwise.} \\ \end{array}\right.$$
It then follows that
\begin{eqnarray}
P\left(c_{12}=c_{22}, c_{11}=c_{21}\right)&=&\sum_{\gamma_1,\gamma_2}
P\left(c_{12}=c_{22} \mid c_{11}=c_{21},\left(\gamma_{1}, \gamma_{2}\right)\right)
P\left(c_{11}=c_{21}\right)P(\gamma_{1})P(\gamma_{2})\nonumber\\
&=&\frac{\alpha^{2}}{M+1}+\frac{(1-\alpha^{2})}{(M+1)^{2}} \label{eq:blah1}
\end{eqnarray}
Similarly,
$$P\left(c_{12}\ne c_{22} \mid c_{11}\ne c_{21},\left(\gamma_{1}, \gamma_{2}\right)\right)=\left\{
\begin{array}{cl}
1 & \mbox{if $\left(\gamma_{1}, \gamma_{2}\right)=(1,1)$} \\
\frac{M}{M+1} & \mbox{otherwise,} \\ \end{array}\right.$$
and proceeding as before, we easily get
\begin{eqnarray}
P\left(c_{12}\ne c_{22}, c_{11}\ne c_{21}\right)&=&\sum_{\gamma_1,\gamma_2}
P\left(c_{12}\ne c_{22} \mid c_{11}\ne c_{21},\left(\gamma_{1}, \gamma_{2}\right)\right)
P\left(c_{11}\ne c_{21}\right)P(\gamma_{1})P(\gamma_{2})\nonumber\\
&=&\frac{M \alpha^{2}}{M+1}+\left(\frac{M}{M+1}\right)^{2}\left(1-\alpha^{2}\right)  \label{eq:blah2}
\end{eqnarray}
The result now easily follows by summing \eqref{eq:blah1} and \eqref{eq:blah2}.

\item[(b)] As before, denote by $\gamma_i=1$ if unit $i\in[m]$ is {\em not} removed from the partition.
By conditioning on $\gamma_1,\gamma_2$ and $c_{11},c_{21}$ we get
$$P\left(c_{12}=c_{22} \mid c_{11}=c_{21},\left(\gamma_{1}, \gamma_{2}\right)\right)=\left\{
\begin{array}{cl}
\frac{6+M}{(M+2)(M+3)} & \mbox{if $\left(\gamma_{1}, \gamma_{2}\right)=(1,1)$} \\
\frac{1}{M+1} & \mbox{otherwise,} \\ \end{array}\right.$$
and proceeding as earlier,
\begin{equation}\label{eq:blah3}
P\left(c_{12}=c_{22}, c_{11}=c_{21}\right)=\frac{(6+M) \alpha^{2}}{(M+1)(M+2)(M+3)}+
\frac{\left(1-\alpha^{2}\right)}{(M+1)^{2}}
\end{equation}
Also,
$$P\left(c_{12}\ne c_{22} \mid c_{11}\ne c_{21},\left(\gamma_{1}, \gamma_{2}\right)\right)=\left\{
\begin{array}{cl}
\frac{M(M+4)+2}{(M+2)(M+3)} & \mbox{if $\left(\gamma_{1}, \gamma_{2}\right)=(1,1)$} \\
\frac{M}{M+1} & \mbox{otherwise,} \\ \end{array}\right.$$
\end{itemize}
from which
\begin{equation}\label{eq:blah4}
P\left(c_{12}\ne c_{22}, c_{11}\ne c_{21}\right)=\left(\frac{M(M+4)+2}{(M+2)(M+3)}\right) \left(\frac{M}{M+1}\right) \alpha^{2}+
\left(\frac{M}{M+1}\right)^{2}\left(1-\alpha^{2}\right).
\end{equation}
The result now easily follows by summing \eqref{eq:blah3} and \eqref{eq:blah4}.
\end{proof}

\subsection{Proof of Proposition 3}
\begin{proof}
Let $P_{C_t} = \{\rho_t \in P : \rho^{\mathfrak{R}_{t}}_{t-1} = \rho^{\mathfrak{R}^{t}}_{t}\}$ denote the
collection of all partitions of the elements of $[m]$ at time $t$ that are
compatible with $\rho_{t-1}$ based on $\bm{\gamma}_t$.  Then by
construction, $\text{Pr}(\rho_t | \bm{\gamma}_t, \rho_{t-1})$ is a random
partition distribution whose support is $P_{C_t}$ so that 
\begin{align*}
\text{Pr}(\rho_t = \lambda | \bm{\gamma}_t, \rho_{t-1}) = \displaystyle \frac{\text{Pr}(\rho_t = \lambda) I[\lambda \in P_{C_t}]}{\sum_{\lambda^\prime} \text{Pr}(\rho_t = \lambda^\prime)I[\lambda^\prime \in P_{C_t}]}.
\end{align*} \qf
It only remains to show that  $\sum_{\lambda \in P_{C_t}} \text{Pr}(\rho_t
= \lambda) = \text{Pr}(\rho_t^{\mathfrak{R}_{t}})$ which is more easily seen
employing cluster label notation.   Let $c_{\gamma_t} = \{ c_{it} :
\gamma_{it}= 0\}$.   By iteratively invoking the sample size consistency
property we have that
\begin{align*}
\text{Pr}(\rho_t^{\mathfrak{R}_{t}}) & =  \sum_{c_{\gamma_t}} \text{Pr}(\rho_t = \{c_{1t}, \ldots, c_{mt}\}) \\
& = \sum_{\lambda \in P_{C_t}} \text{Pr}(\rho_t = \lambda ),
\end{align*}
where the last equality holds since  summing over $c_{\gamma_t}$ is based
only  on cluster labels that are not fixed from time point $t-1$ to $t$
which results in summing over all possible compatible partitions (i.e.,
$\lambda \in P_{C_t}$).
\end{proof}

\section{Details Associated With the MCMC Algorithm}\label{MCMC}
Here we provide much more detail associated with the MCMC scheme. We place
emphasis on the updating steps for $\bm{\gamma}_{it}$ and $\rho_t$ as
the other steps are straightforward once these  parameters have been
updated. That said,  pseudocode of the entire algorithm is provided in
Algorithm \ref{MCMC.algorithm}.    A key component of the MCMC algorithm
is to check the compatibility between $\rho_{t-1}$ and $\rho_t$, and
between $\rho_{t}$ and $\rho_{t+1}$ when updating $\bm{\gamma}_t$ and
$\rho_t$.  This is equivalent to ensuring that
$\rho^{\mathfrak{R}_t}_{t-1} = \rho^{\mathfrak{R}_t}_t$ and
$\rho^{\mathfrak{R}_{t+1}}_{t} = \rho^{\mathfrak{R}_{t+1}}_{t+1}$.   We
describe the process of updating each of the $c_{it}$ and $\gamma_{it}$
sequentially in time so that the entire vector $\bm{\gamma}_t$ is updated first and then $\bm{c}_t$.

\subsection{Updating $\gamma_{it}$} \label{gamma.update}
First note that $\bm{\gamma}_t$ only connects $\rho_{t-1}$ to $\rho_t$ so
that when updating $\gamma_{it}$ only compatibility between $\rho_{t-1}$
and $\rho_t$ needs to be checked (i.e., $\rho_t$ remains compatible with
$\rho_{t+1}$ due to $\gamma_{t+1}$ by construction).   We detail updating
$\gamma_{it}$ in an MCMC algorithm based on its full conditional found
in \eqref{full.conditional.gamma1b} of the main paper and which we provide
here for sake of clarity
\begin{align}\label{gamma.full.conditional}
\text{Pr}(\gamma_{it} = 1 \mid  -) = \displaystyle\frac{\alpha_t}{\alpha_t  + (1-\alpha_t)\text{Pr}(\rho^{\mathfrak{R}^{(+i)}_{t} }_t)/\text{Pr}(\rho^{\mathfrak{R}^{(-i)}_{t} }_t)}\text{I}[\rho^{\mathfrak{R}^{(+i)}_{t}}_{t-1} = \rho^{\mathfrak{R}^{(+i)}_{t}}_{t}].
\end{align}
The appeal of this form of the full conditional compared to that found in
\eqref{full.conditional.gamma1} of the main paper is that the EPPF used to
compute $\text{Pr}(\rho^{\mathfrak{R}^{(+i)}_{t} }_t)$ and
$\text{Pr}(\rho^{\mathfrak{R}^{(-i)}_{t} }_t)$ need not have a tractable
normalizing constant.  That said,  an exchangeable sequence of cluster
labels ($c_{1t}, \ldots c_{mt}$) is necessary.    Now let
$\gamma_{it}^{(d)}$ and $\rho^{(d)}_t$ be the value of $\gamma_{it}$ and
$\rho_t$ at the $d$th MCMC iterate.  Note that there are four scenarios to
consider when moving from   $\gamma_{it}^{(d-1)}$ to $\gamma_{it}^{(d)}$.
They are
\begin{enumerate}
\item  $\gamma_{it}^{(d-1)} = 1 \rightarrow \gamma_{it}^{(d)} = 0$ (For this move $\rho^{(d)\mathfrak{R}^{(-i)}_{t} }_{t-1}  = \rho^{(d-1)\mathfrak{R}^{(-i)}_{t} }_t$ continues to hold),
\item  $\gamma_{it}^{(d-1)} = 1 \rightarrow \gamma_{it}^{(d)} = 1$ (For this move $\rho^{(d)\mathfrak{R}^{(+i)}_{t} }_{t-1}  = \rho^{(d-1)\mathfrak{R}^{(+i)}_{t} }_t$ continues to hold),
\item $\gamma_{it}^{(d-1)} = 0 \rightarrow \gamma_{it}^{(d)} = 0$  (For this move $\rho^{(d)\mathfrak{R}^{(-i)}_{t} }_{t-1}  = \rho^{(d-1)\mathfrak{R}^{(-i)}_{t} }_t$ continues to hold), and
\item $\gamma_{it}^{(d-1)} = 0 \rightarrow \gamma_{it}^{(d)} = 1$  (For this move $\rho^{(d)\mathfrak{R}^{(+i)}_{t} }_{t-1}  = \rho^{(d-1)\mathfrak{R}^{(+i)}_{t} }_t$ needs to be verified).
\end{enumerate}
Thus compatibility needs to be checked only for (4).







As expected, calculating the ratio $\text{Pr}(\rho^{\mathfrak{R}^{(+i)}_{t} }_{t})/\text{Pr}(\rho^{\mathfrak{R}^{(-i)}_{t} }_{t})$ in \eqref{gamma.full.conditional} is the most challenging part of computing $\text{Pr}(\gamma_{it} = 1 \mid  -)$.  However, it can be straightforwardly calculated using exchangeability and ideas from \cite{neal:2000}.  Under the assumption of exchangeability, which permits allocating the $i$th unit by treating it as if it were the last unit,  we have that
\begin{align} \label{probability.ratio}
\frac{\text{Pr}(\rho^{\mathfrak{R}^{(+i)}_{t} }_{t})}{\text{Pr}(\rho^{\mathfrak{R}^{(-i)}_{t} }_{t})} = \frac{\text{Pr}(c_{it} | \rho^{\mathfrak{R}^{(-i)}_{t} }_{t}) \text{Pr}(\rho^{\mathfrak{R}^{(-i)}_{t} }_{t})}{\text{Pr}(\rho^{\mathfrak{R}^{(-i)}_{t} }_{t})} = \text{Pr}(c_{it} | \rho^{\mathfrak{R}^{(-i)}_{t} }_{t}).
\end{align}
Note that the probability in \eqref{probability.ratio} is a standard
calculation, used in Neal's Algorithm 8, for each $c_{it}$. When the EPPF
does not have a tractable normalizing constant, one may compute the
unnormalized probability of allocation to each of the $k_t$ existing
clusters and to a new singleton cluster and then normalize to obtain
(S.1).  Of course, here we know the value of $c_{it}$ from
$\rho_t^{(d-1)}$ and, in constraints to Neal's Algorithm 8, this
computation is done for the sake of computing the full conditional
distribution of $\gamma_{it}$.  Once \eqref{probability.ratio} is
calculated,  computing \eqref{gamma.full.conditional} and updating
$\gamma_{it}$ is straightforward.


\subsection{Updating $\rho_{t}$ Using Cluster Labels} \label{rho.update}

First note that only those $c_{it}$ that correspond to $\gamma_{it} = 0$
are updated. As a result,  the compatibility between $\rho_{t-1}$ and
$\rho_{t}$ is preserved and so only the compatibility between $\rho_t$ and
$\rho_{t+1}$ needs to be checked when updating any of the $c_{it}$.
Recall that the full conditional of $c_{it}$ corresponding to $\gamma_{it}
= 0$ is
\begin{align}\label{full.conditional.rho}
\text{Pr}(c_{it} = h \mid  -) \propto
\left\{
\begin{array}{cl}
N(Y_{it} \mid  \mu^{\star}_{c_{it} = h, t}, \sigma^{2\star}_{c_{it} = h, t})\text{Pr}(c_{it} = h)\text{I}[\rho^{h\mathfrak{R}_{t+1}}_{t} = \rho^{\mathfrak{R}_{t+1}}_{t+1}]  & \mbox{for $h = 1, \ldots, k_t^{-i}$, }   \\
N(Y_{it} \mid  \mu^{\star}_{new_h, t}, \sigma^{2\star}_{new_h, t})\text{Pr}(c_{it} = h)\text{I}[\rho^{h\mathfrak{R}_{t+1}}_{t} = \rho^{\mathfrak{R}_{t+1}}_{t+1}] &    \mbox{for $h = k_t^{-i}+1$,}
\end{array}
\right.
\end{align}
where $\text{Pr}(c_{it} = h) = \text{Pr}(c_{1t}, \ldots, c_{it} = h,
\ldots, c_{mt})$, and $k_t^{-i}$  is the number of clusters at time
$t$ when the $i$th unit has been removed.  The partition constructed from $\{c_{1t}, \ldots,
c_{it} = h, \ldots, c_{mt}\}$ is denoted as $\rho^h_t = \{ S_{1t}^{-i},
\ldots, S_{ht}^{-i}\cup \{i\}, \ldots, S_{k_t^{-i} t}^{-i}\}$ where
$S^{-i}_{jt}$ denotes the $j$th cluster at time $t$ with the $i$th unit
removed.  Note that it is possible that  $S^{-i}_{jt} = S_{jt}$. Further, abusing notation,  for $h = k_t^{-i}+1$ we have $\rho^h_t = \{ S_{1t}^{-i},
\ldots, S_{ht}^{-i}, \ldots, S_{k_t^{-i} t}^{-i}, \{i\}\}$.  Additionally,
$\mu^{\star}_{new_h, t}$ and $\sigma^{2\star}_{new_h, t}$ are auxiliary
parameters drawn from the prior as in \cite{neal:2000}'s Algorithm 8 (with
one auxiliary parameter). Then based on a spatial product
partition model for $\rho_t$, for $h = 1, \ldots, k_t^{-i}$ the $\text{Pr}(c_{it} = h)$ becomes
\begin{align}\label{partition.probability}
 \text{Pr}(c_{it} = h) = Pr(\rho_t^h)  \propto M\Gamma(|S^{-i}_{ht} \cup \{i\}|) g(\bm{s}^{-i\star}_{ht} \cup \bm{s}_i  | \nu_0) \prod_{j\ne h}^{k_t^{-i}} M\Gamma(|S^{-i}_{jt} |) g(\bm{s}^{-i\star}_{jt}  | \nu_0),
\end{align}
while for $h=k_t^{-i} + 1$
\begin{align}\label{partition.probability2}
 \text{Pr}(c_{it} = h) = Pr(\rho_t^h)  \propto M\Gamma(|\{i\}|) g(\bm{s}_i  | \nu_0) \prod_{j = h}^{k_t^{-i}} M\Gamma(|S^{-i}_{jt} |) g(\bm{s}^{-i\star}_{jt}  | \nu_0),
\end{align}
where $g(\cdot)$ is the auxiliary similarity function detailed in
\cite{page&quintana:2016} and $\bm{s}^{-i\star}_{jt} =\{\bm{s}_{i'} : i'
\in S^{-i}_{jt}\}$ are the spatial coordinates from units that belong to
the $j$th cluster at time $t$.  Updating $c_{it}$ can be carried out by
evaluating \eqref{full.conditional.rho} based on
\eqref{partition.probability} or \eqref{partition.probability2} for each $h = 1, \ldots, k_t^{-i} + 1$ and
then normalizing.

Once each of the $\bm{c}_t$ and $\bm{\gamma}_{t}$ are updated the MCMC
algorithm is completed by cycling through remaining model and latent
parameters found in model \eqref{FullModel} and updating them on an
individual basis using well known approaches.  In order to visualize all the moving parts of the MCMC algorithm we provide some pseudocode in Algorithm \ref{MCMC.algorithm}.  For sake of simplicity,  Algorithm \ref{MCMC.algorithm} describes an MCMC procedure that can be employed to sample from the joint posterior distribution based on model \eqref{dat.gen.mech}.

\begin{algorithm}
  \footnotesize
  \caption{: \footnotesize Pseudocode for the MCMC algorithm  for model \eqref{dat.gen.mech} of main article.  Let
  $T$ be the number of time points,  $m$ the number of units at each time point, and
  $D$ the number of MCMC iterations.}
  \label{MCMC.algorithm}
  \begin{algorithmic}[1]
    \For{$d = 1,\ldots,D$}   
        \For{$t = 1, \ldots, T$}\Comment{For each $t$, update the entire $\bm{\gamma}_{t}$ vector first and then $\bm{c}_t$}
            \For{$ i = 1, \ldots, m $}
              \State Set $\gamma^{(d)}_{i1} = 0$.
              \If{$t > 1$}  
              \State - Update $\gamma_{it}$ based on procedure described in Section \ref{gamma.update}.  That is, 
                \If{$\gamma^{(d-1)}_{it} = 1$} 
                 \State  Move to $\gamma^{(d)}_{it}$ using Bernoulli probability in \eqref{gamma.full.conditional}.  Compatibility holds by construction. 
                \EndIf
                \If{$\gamma^{(d-1)}_{it} = 0$} 
                 \State  Move to $\gamma^{(d)}_{it}$ using Bernoulli probability in \eqref{gamma.full.conditional}.  If $\gamma^{(d)}_{it}=0$, then compatibility 
                 \State \ \ \ holds by construction. If $\gamma^{(d)}_{it}=1$, the compatibility needs to be checked.  If 
                 \State \ \ \ $\rho^{(d)\mathfrak{R}^{(+i)}_{t} }_{t-1}  \ne \rho^{(d-1)\mathfrak{R}^{(+i)}_{t} }_t$, then set $\gamma^{(d)}_{it}=0.$
                \EndIf
              \EndIf
            \EndFor    
            \For{$ i = 1, \ldots, m $}
                \State - Update $c_{it}$ based on procedure described in Section \ref{rho.update}.
                \For{$ h = 1, \ldots, k^{-i}_t $}
                \State Compute the unnormalized multinomial probability in \eqref{full.conditional.rho} based on \eqref{partition.probability}. 
                \If{$\rho^{h\mathfrak{R}_{t+1}}_{t} \ne \rho^{\mathfrak{R}_{t+1}}_{t+1}$}
                \State Set unnormalized multinomial probability to zero.
                \EndIf
                \EndFor
                \For{$h = k^{-i}_t +1$}
                \State Compute the unnormalized multinomial probability in \eqref{full.conditional.rho} based on \eqref{partition.probability2}. 
                \EndFor
               \State Sample $c_{it}$ using the normalized   $k^{-i}_t +1$ multinomial probabilities.
            \EndFor    
            \For{$ j = 1, \ldots, K^{(d)} $} \Comment{$K^{(d)}$ = number of clusters at the $d$th iteration.}
                \State - Update $\mu^{\star}_{jt}$ based on Gaussian full conditional derived using well known arguments.
                \State - Update $\sigma^{2\star}_{jt}$ using a random walk Metropolis step.
            \EndFor 
          \State - Update $\theta_{t}$ based on Gaussian full conditional derived using well known arguments.
          \State - Update $\alpha_{t}$ based on beta full conditional  derived using well known arguments .
      \EndFor
      \State - Update $\tau^2$ using a random walk Metropolis step.
      \State - Update $\phi_0$ based on Gaussian full conditional derived using well known arguments.
      \State - Update $\lambda^2$ using a random walk Metropolis step.
    \EndFor                     
  \end{algorithmic}
\end{algorithm}

\section{Simulation Studies}\label{simulation}

In this section we provide more details associated with Simulation 1,  the competitors included in Simulation 3, and provide results from additional simulations similar to that described in the Section 3.3 of the main document.  We then provide details associated with a simulation study that includes spatial information.

\subsection{Simulation 1: Continued}

Table \ref{sampleARItrue} contains the adjusted Rand index values between the estimated partitions and that which was used to generate data.  Interestingly, as $\alpha$ increases, the ARI values also tend to increase.  

\spacingset{1}
\begin{table}[t]
\caption{Adjusted Rand index when comparing $\hat{\rho}_t$ to
the true $\rho_t$ for $t = 1, \ldots, 5$.   Note that $ARI(\cdot,
\cdot)$ denotes the adjusted Rand index as a function of two partitions.  These values are averaged over the 100
generated data sets. The values in parenthesis are Monte Carlo standard
errors. }
\vspace{0.25 cm} \centering
\begin{tabular}{l c c c c c}
  \toprule
 & $ARI(\hat{\rho}_1, \rho_1)$ &  $ARI(\hat{\rho}_2, \rho_2)$  &  $ARI(\hat{\rho}_3, \rho_3)$&  $ARI(\hat{\rho}_4, \rho_4)$ &  $ARI(\hat{\rho}_5, \rho_5)$ \\
  \midrule
$\alpha=0.0$ 		& 0.58 (0.03) & 0.63 (0.03) & 0.58 (0.03) & 0.54 (0.03) & 0.56 (0.03) \\
$\alpha=0.1$ 		& 0.63 (0.03) & 0.56 (0.03) & 0.55 (0.03) & 0.62 (0.03) & 0.57 (0.03) \\
$\alpha=0.25$ 		& 0.52 (0.03) & 0.57 (0.03) & 0.55 (0.03) & 0.63 (0.03) & 0.62 (0.03) \\
$\alpha=0.5$ 		& 0.60 (0.03) & 0.70 (0.03) & 0.69 (0.03) & 0.66 (0.02) & 0.59 (0.03) \\
$\alpha=0.75$ 		& 0.78 (0.02) & 0.77 (0.02) & 0.82 (0.02) & 0.80 (0.02) & 0.75 (0.02) \\
$\alpha=0.9$ 		& 0.83 (0.02) & 0.86 (0.02) & 0.87 (0.02) & 0.84 (0.02) & 0.76 (0.02) \\
$\alpha=0.9999$ 	& 0.92 (0.01) & 0.92 (0.01) & 0.92 (0.01) & 0.92 (0.01) & 0.92 (0.01) \\
   \bottomrule
\end{tabular}
\label{sampleARItrue}
\end{table}%
\spacingset{1.5}

\subsection{Simulation 3: Continued} \label{simstudy3_cont}

As referenced in the main article, Figure \ref{synthetic_data1} provides an example of the type of data that is generated in the simulation of Section \ref{sim.study3} in the main article.  To each of the 100 data sets generated, we fit our method and four other procedures.  We know provide specific details of the competing methods.
\begin{figure}[htbp]
\begin{center}
\includegraphics[scale=0.75]{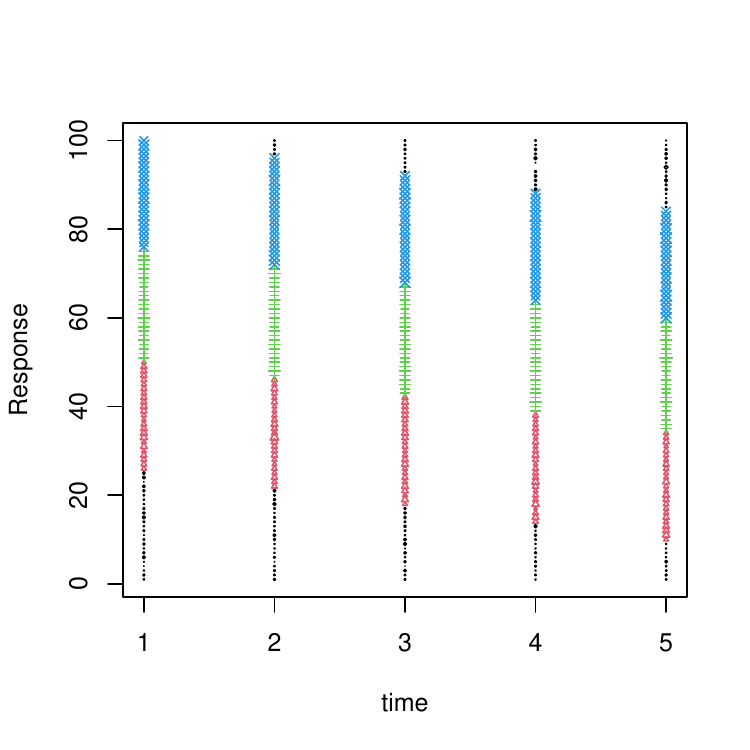}
\caption{One realization of a synthetic data set from simulation study in
Section \ref{sim.study3}.  Each color corresponds to a cluster and the
size of plotted symbol is proportional to the value of point being
plotted.} \label{synthetic_data1}
\end{center}
\end{figure}

\begin{enumerate}
\item weighted DDP (wddp):  A complete description of this model can be found in Section 4 of  \cite{quintana2020dependent}
   (and \citealt{BNPbook:2015}).  We only provide pertinent details here.  Let $\bm{z}_i = (Y_i, t_i)$ be the response and time pair for $i=1, \ldots, mT$.  The wddp models $\bm{z}_i$ with a Dirichlet process mixture model (DPM) and then derives the  conditional model $(Y_i | t_i)$.  In hierarchical form (including cluster labels) the model is 
\begin{align*}
\bm{z}_i | \bm{\mu}^{\star}, \bm{\Sigma}^{\star}, c_i & \stackrel{ind}{\sim} N_2(\bm{\mu}^{\star}_{c_i}, \bm{\Sigma}^{\star}_{c_i}), \ \ i = 1, \ldots, mT \\
\bm{\mu}^{\star}_j & \stackrel{iid}{\sim} N_2(\bm{\mu}_0, \bm{\Sigma}_0), \ \ j = 1, \ldots, K\\
\bm{\Sigma}^{\star}_j & \stackrel{iid}{\sim} \text{Inverse-Wishart}_2(\nu, \psi\bm{I}), \ \ j = 1, \ldots, K\\
\bm{\mu}_0 & \stackrel{iid}{\sim} N_2(\bm{m}, s^2\bm{I}),\\
\bm{\Sigma}_0 & \stackrel{iid}{\sim} \text{Inverse-Wishart}_2(\nu_0, \psi_0\bm{I}),\\
\text{Pr}(c_i = j) & = \pi_j  \ \ \mbox{where $\pi_j = V_j \prod_{\ell < j} (1-V_{\ell})$} ,\\
V_{\ell} & \stackrel{iid}{\sim} \text{Beta}(1, M).
\end{align*}
We set $\bm{m} = 0\bm{j}$,  $s^2 = 25$, $\nu = \nu_0 = 4$,  $\psi = \psi_0 = 1$, $K=30$, and $M = 1$.  
The model  induces a weight-dependent mixture model of regressions
\begin{align*}
f(y_i | t_i) = \sum_{j=1}^{K} w_j(t_i) N(y_i | \beta^{\star}_{0j} + \beta^{\star}_{1j} t_i, \sigma^{2\star}_j),
\end{align*}
where
\begin{align*}
w_j(t_i) =  \frac{\pi_j N(t_i | \mu^{\star}_{2j}, {\Sigma}_{22j}^{\star})}{\sum_{\ell=1}^{K} \pi_{\ell} N(t_i | \mu^{\star}_{2\ell},  {\Sigma}_{22\ell}^{\star})}, \ \ j = 1, \ldots, K, 
\end{align*}
and $\beta^{\star}_{0j}  = \mu^{\star}_{1j} - \frac{\Sigma^{\star}_{12j}}{\Sigma^{\star}_{22j}}\mu^{\star}_{2j}$, $\beta_{1h} =  \frac{\Sigma^{\star}_{12j}}{\Sigma^{\star}_{22j}}$, and $\sigma^{2\star}_{j} = \Sigma^{\star}_{11j} -  \frac{\Sigma^{\star}_{12j}\Sigma^{\star}_{21j}}{\Sigma^{\star}_{22j}}$.   
Note that time is include in the weights of the the conditional model which is employed to calculate LPML and WAIC.  A blocked Gibbs sampler was employed to sample from the posterior where $V_{K} =1$ to ensure that $\sum_{j=1}^K \pi_j = 1$.

\item linear DDP (lddp): A complete description of this model is provided in chapter 4.4.2 of
    \cite{BNPbook:2015} (see also \citealt{quintana2020dependent}).  We only provide pertinent details here.  As with the wddp model, for the lddp we consider $(Y_i, t_i)$  for $i=1, \ldots, mT$.  Time is incorporated in the atoms of a Dirichlet process (DP) so that the $j$th atom is expressed as $\sum_{\ell}^qB_{\ell}(t, \bm{\xi})\beta_{j\ell}$ where  $B_{\ell}(t, \bm{\xi})$ denotes the $\ell$-th B-spline basis function evaluated at $t$ for knots $\bm{\xi}$.  Letting $\bm{\beta} = (\beta_1, \ldots, \beta_q)$ and $\bm{B}(t_i, \bm{\xi})$ the $q$-dimensional B-spline basis vector for unit $i$ and after introducing cluster labels, the lddp model can be expressed hierarchically as
\begin{align*}
Y_i | \bm{\beta}^{\star}, {\sigma}^{2\star}, c_i & \stackrel{ind}{\sim} N(\bm{B}'(t_i, \bm{\xi})\bm{\beta}^{\star}_{c_i}, {\sigma}^{2\star}_{c_i}), \ \ i = 1, \ldots, mT,\\
\bm{\beta}^{\star}_j & \stackrel{iid}{\sim} N_q(\bm{\beta}_0, \bm{\Sigma}_0), \ \ j = 1, \ldots, k,\\
\sigma^{2\star}_j & \sim \text{Inverse-Gamma}(a, b), \ \ j = 1, \ldots, k, \\
\bm{\beta}_0 & \stackrel{iid}{\sim} N_2(\bm{m}, s^2\bm{I}),\\
\bm{\Sigma}_0 & \stackrel{iid}{\sim} \text{Inverse-Wishart}_2(\nu_0, \psi_0\bm{I}),\\
\{c_1, \ldots, c_{mT}\} & \sim CRP(M).
\end{align*}
We set $\bm{m} = 0\bm{j}$,  $s^2 = 25$, $\nu_0 = q+2$,  $\psi_0 = 10$, $a=b=1$, and $M = 1$.  Neal's Algorithm 8 (\citealt{ neal:2000}) was employed to sample from the posterior distribution.

\item Griffiths-Milne dependent Dirichlet process (gmddp) mixture.  This is carried out using {\tt DDPdensity} in the {\tt BNPmix} package found in {\tt R}.  The function considers partially exchangeable data (Lijoi et al., 2014) such that exchangeability is assumed within each group and the vector of random probability measures at each time point are modeled jointly as a vector of GM-DDP.

\item A temporally independent $CRP(M)$ model (ind\_crp): This model is a special case of \cite{Caron:2007}'s and our model.  Specifically,  $\alpha$ is set to 0.  For this procedure,  the following model was fit separately for each time period.
\begin{align*}
Y_{i} \mid  \bm{\mu}^{\star}, \bm{\sigma}^{2\star}, c_{i} & \stackrel{ind}{\sim} N(\mu^{\star}_{c_{i}}, \sigma_{c_{i}}^{2\star}), \ i = 1, \ldots, m,   \\
(\mu_{j}^{\star}, \sigma^{\star}_{j})\mid \theta, \tau^2 & \stackrel{ind}{\sim} N(\theta, \tau^2) \times UN(0,A_{\sigma}), \ j = 1, \dots, k ,\\
(\theta, \tau) & \stackrel{iid}{\sim} N(m_0, s_0^2) \times UN(0, A_{\tau}).
\end{align*}
We set $m_0=0$, $s_0^2=10^2$, $A_{\sigma}=0.5sd(Yvec)$, and  $A_{\tau}=100$.  Neal's Algorithm 8 (\citealt{ neal:2000}) was used to sample from the posterior distribution.

\item A temporally static $CRP(M)$ model (static\_crp): This procedure is a special case of  \cite{Caron:2007} ($\alpha = 1$) and is fit to a concatenated version of the data  $Y_i$, for $i = 1, \ldots, mT$.   Specifically, the following model was fit
\begin{align*}
Y_{i} \mid  \bm{\mu}^{\star}, \bm{\sigma}^{2\star}, c_{i} & \stackrel{ind}{\sim} N(\mu^{\star}_{c_{i}}, \sigma_{c_{i}}^{2\star}), \ i = 1, \ldots, mT,   \\
(\mu_{j}^{\star}, \sigma^{\star}_{j})\mid \theta, \tau^2 & \stackrel{ind}{\sim} N(\theta, \tau^2) \times UN(0,A_{\sigma}), \ j = 1, \dots, k ,\\
(\theta, \tau) & \stackrel{iid}{\sim} N(m_0, s_0^2) \times UN(0, A_{\tau}).
\end{align*}
We set $m_0=0$, $s_0^2=10^2$, $A_{\sigma}=0.5sd(Yvec)$, and  $A_{\tau}=100$.  Neal's Algorithm 8 (\citealt{ neal:2000}) was used to sample from the posterior distribution.

\end{enumerate}

\subsubsection{Results from Additional Synthetic Data}

In addition to the synthetic data generated in Simulation 3 of the main document,  we also generated data as described in the following two scenarios.
\begin{enumerate}

\item[] {\bf Scenario 1}:  For the $i$th unit, we employ the following as a data generating mechanism
\begin{align*}
y_{it} =  \omega y_{it-1} + \epsilon_{it}, \ \mbox{for $i=1, \ldots, m$, and $t = 1, \ldots, T$},
\end{align*}
where  $|\omega| < 1$ and $\epsilon_{it} \sim N(0, v^2)$. For this
scenario measurements are are correlated across time, but independent
between the $m$ units.   Data were generated with $m=100$ and using
the following levels of factors of interest
\begin{itemize}
\item $\omega \in \{0, 0.1, 0.25, 0.5, 0.75, 0.9\}$
\item $v^2 \in \{0.5^2, 1^2\}$
\item $T \in \{5, 10\}$
\end{itemize}
Notice that for this scenario, there is no ``true'' partition and so we are interested only in comparing the model fit of our approach to that of the five competitors.


\item[] {\bf Scenario 2}:  For the $i$th unit, we employ the following as a data generating mechanism
\begin{align*}
y_{it} =  \omega_{c_i} y_{it-1} + \epsilon_{it}, \ \mbox{for $i=1, \ldots, m$, and $t = 1, \ldots, T$},
\end{align*}
where as before $c_{it} \in \{1,2,3,4\}$ with $\omega \in \{-0.75,
-0.25, 0.25, 0.75\}$ and $\epsilon_{it} \sim N(0, v^2)$.  As in the
previous scenarios measurements are  correlated across time, but
independent between the $m$ units.   Data were generated with $m=100$
and using the following levels of factors of interest
\begin{itemize}
\item $v^2 \in \{0.5^2, 1^2\}$
\item $T \in \{5, 10\}$
\end{itemize}
In this scenario, there is a ``true'' partition but our approach, nor the competitors, are
parametrized in such a way as to detect it.  Indeed, our method models
temporal dependence only through the partition (i.e, there is no
temporal correlation parameter in the likelihood).  That said, we still compare partition recovery by way of the adjusted Rand index (ARI) in addition to model fit. 

\end{enumerate}

In both scenarios the function {\tt arima.sim} in {\tt R} (\citealt{R:2020}) is used to generate
the 100 replicate data sets.  We collect 1,000 MCMC iterates after discarding the first 25,000 as
burn-in and thinning by 25 (i.e., 50,000 total MCMC draws were
collected).    The prior parameters that we used are $A_{\sigma} =
0.5sd(vec(Y))$, $A_{\tau} = 100$, $A_{\lambda} = 100$, $m_0 = 0$, $s2_0 =
100$, $a_{\alpha} = 1$, $b_{\alpha} = 1$.  WAIC is used to compare each of the procedures in terms of model fit
and ARI to compare ability of estimating the true
partition structure.  Results are found in Figures  \ref{ss_dt1_waic} - \ref{ss_dt4_ari}.  From Figure \ref{ss_dt1_waic} we see that our approach produces the smallest WAIC for all factors considered Scenarios 1's data generating schemes.  Thus, our approach tends to fit the data best.

\begin{figure}[htbp]
\begin{center}
\includegraphics[scale=0.75]{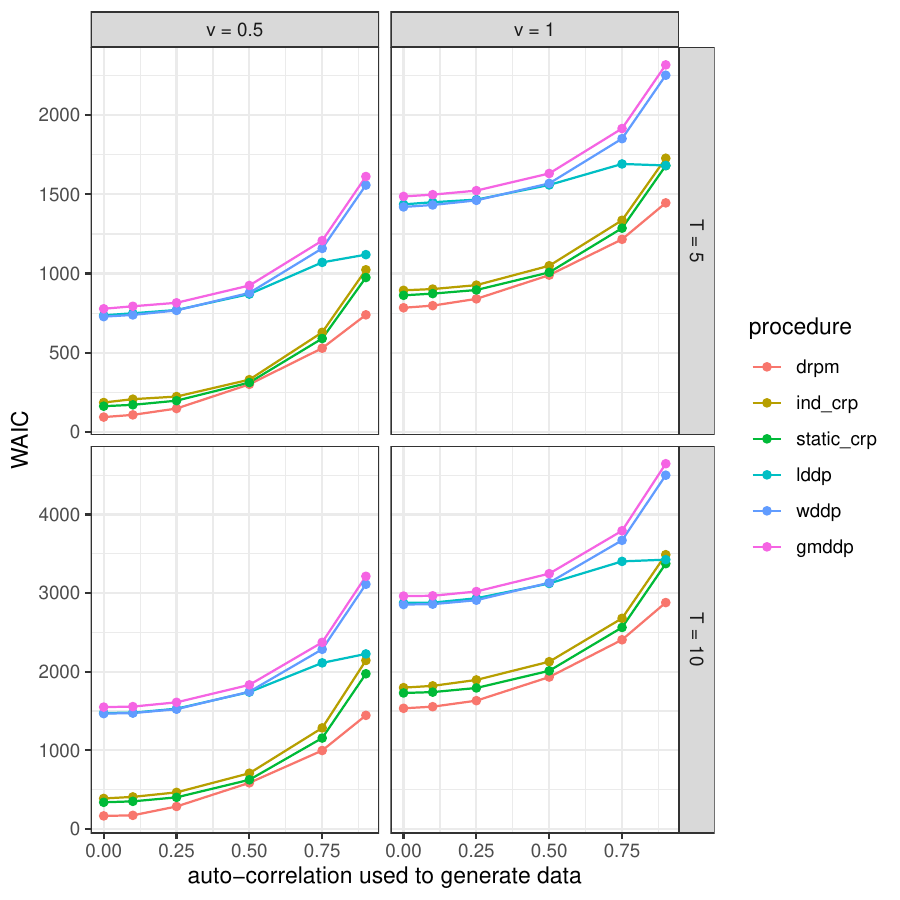}
\caption{Results for WAIC from the first data generating scenario.} \label{ss_dt1_waic}
\end{center}
\end{figure}

For Scenario 2,  static\_crp is quite competitive to our approach and produces similar WAIC values.  Apart from that, our approach does better than the other competitors.  From Figure \ref{ss_dt4_ari} our approach does much better at recovering the partition compared to other procedures.  That said, we the ARI values are still quite small (which was expected).

\begin{figure}[htbp]
\begin{center}
\includegraphics[scale=0.75]{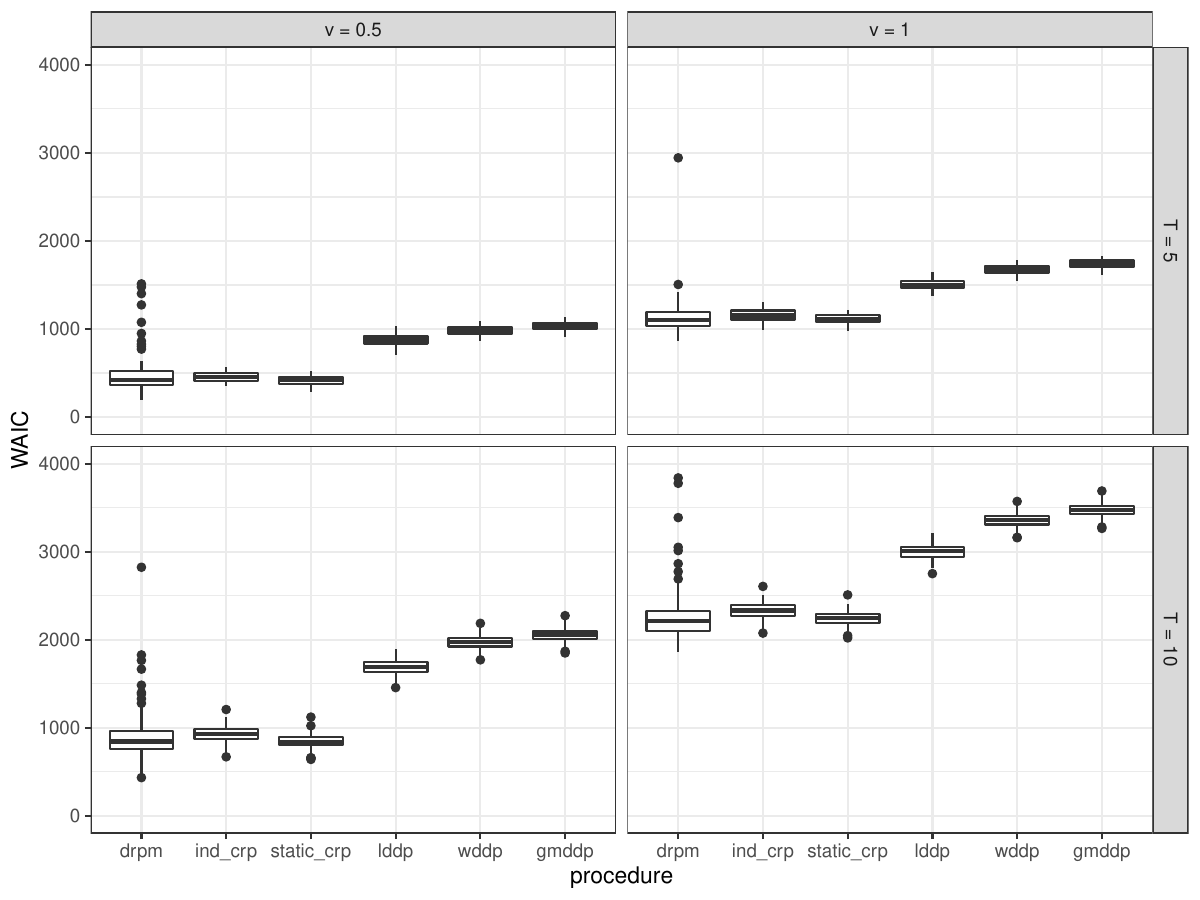}
\caption{Results for WAIC from the fourth data generating scenario.} \label{ss_dt4_waic}
\end{center}
\end{figure}

\begin{figure}[htbp]
\begin{center}
\includegraphics[scale=0.75]{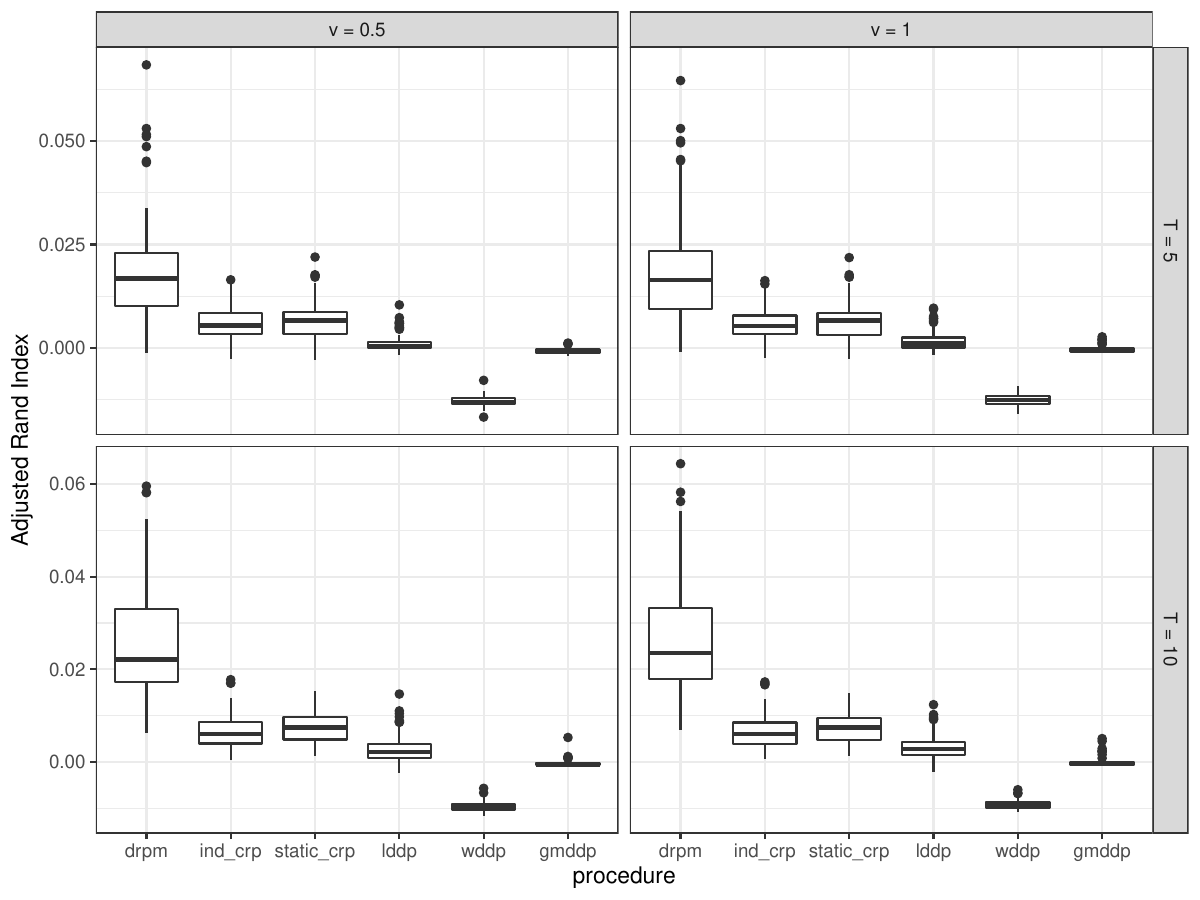}
\caption{Results for ARI from the fourth data generating scenario.} \label{ss_dt4_ari}
\end{center}
\end{figure}

\subsection{Simulation 4: Space-Time Data Generation} \label{simstudy4}
Here we discuss our final simulation study, where we investigated the
performance of our procedure when both space and time are considered. To
do so, we created synthetic data sets that contain spatio-temporal
structure.  Each  employs a $15 \times 15 $ regular grid with spatial
locations coming from the unit interval. In addition,  either 5 or 10 time
points were considered resulting in 1,125 or 2,250 total observations.
Response values were generated in two  ways.  The first employs a Gaussian
process with a separable spatio-temporal exponential covariance function.
We set the spatial scale to 0.3, the temporal scale to 2 and the sill to
1.75 (see \citealt{padoan&bevilacqua:2015} for more details).   Note that
no ``true'' partition exists for this data generating mechanism.  However,
we study it to explore performance of our method when spatial structure
exists among observations but was not induced through partitioning.   The
second method of generating response values essentially employs model
\eqref{dat.gen.mech} as a data generating mechanism.  Spatio-temporal
partitions were generated using \eqref{similarity.function} together with
conditional cluster label probabilities of  \citet[pg.
265]{PPMxMullerQuintanaRosner}  and setting $\alpha_t=\alpha$ for all $t$
with $\alpha \in \{0, 0.5, 0.9\}$ (note that for $\alpha = 0$ no temporal
dependence exists among partitions).   In the similarity function
\eqref{similarity.function} we considered $\nu_0 \in \{2,20\}$ where
$\nu_0 = 2$ corresponds to light weight on spatial proximity and $\nu_0 =
20$ moderate weight.   Finally, we set $\tau^2=1$ and
$\sigma^{2\star}_{c_{it}t} = \sigma^2 = 0.04$ for all $i$ and $t$
resulting in smaller with-in cluster variability relative to
between-cluster variability.

To determine the impact that each component of our spatio-temporal
partition model has on model fit, we fit the hierarchical model
\eqref{dat.gen.mech} to each synthetic data set using a variety of random
partition models which are listed below. As a competitor, we consider  a
linear dependent Dirichlet process \citep{MacEachern:2000,deIorio:2009},
indexing the random probability measure through the mean function of the
atoms by space and time.  To ensure sufficient flexibility, B-spline basis
functions for both spatial coordinates were employed.  The details of each
model considered are
\begin{enumerate}
\item[] Model 1: $(\rho_1, \ldots, \rho_T) \sim  stRPM(\bm{\alpha}, \nu_0, M)$
\item[] Model 2: $\rho_t \stackrel{iid}{\sim}  sPPM(\nu_0, M)$ for $t = 1, \ldots, T$.
\item[] Model 3: $(\rho_1, \ldots, \rho_T) \sim  tRPM(\bm{\alpha}, M)$
\item[] Model 4: $\rho_t \stackrel{iid}{\sim}  CRP(M)$ for $t = 1, \ldots, T$.
\item[] Model 5: linear dependent Dirichlet process  mixture model (DDPM).
\end{enumerate}
Additionally,  for each model that employs the sPPM, we considered both
$\nu_0 = 2$  (models 1a, 2a) and $\nu_0 = 20$ (models 1b, 2b). For each
data generating scenario, 100 data sets were created  and each of the
models listed was  fit by collecting 1,000 MCMC samples after discarding
the first 5,000 as burn-in and thinning by 5 after setting $A_{\sigma} =
1$ and $A_{\tau} = 2$.  Model fits were compared using WAIC. Results can
be found in Figures \ref{waicGP} and \ref{waicDRP}.

\spacingset{1}
\begin{figure}[htbp]
\begin{center}
\includegraphics[scale=0.6]{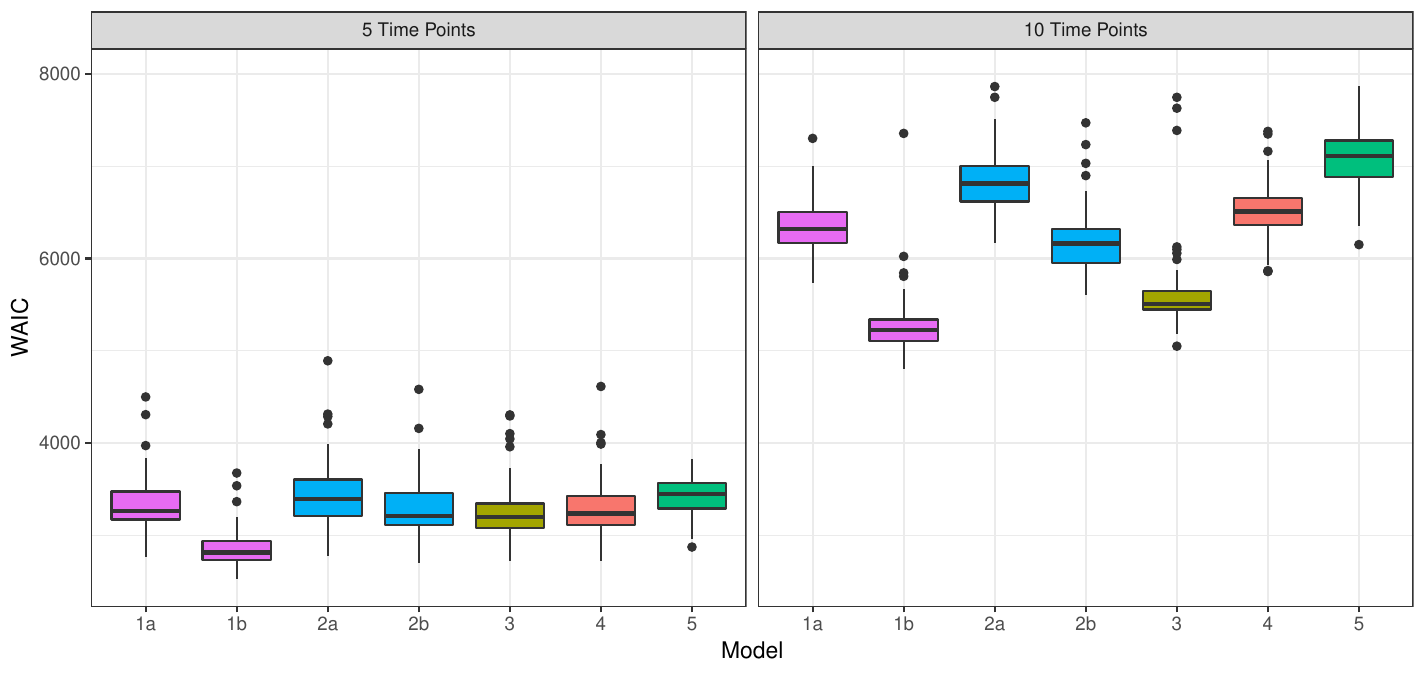}
\caption{Results from simulation study when observations were generated using a
spatio-temporal Gaussian process. Boxplots display the 100 WAIC values
that correspond to model fit for each synthetic data set.  Note that
smaller WAIC values indicate a better fit. } \label{waicGP}
\end{center}
\end{figure}
\spacingset{1.5}

\spacingset{1}
\begin{figure}[htbp]
\begin{center}
\includegraphics[scale=0.6]{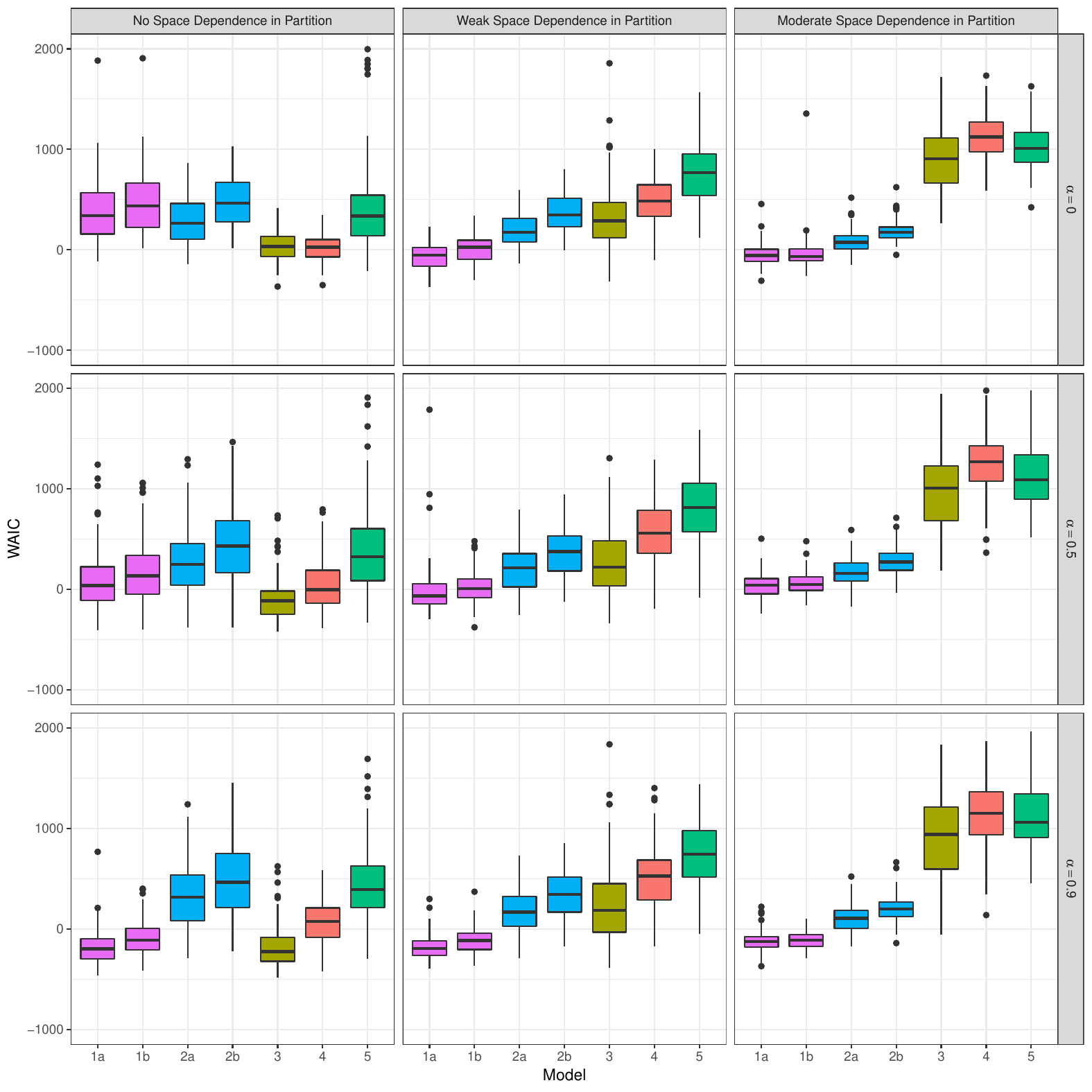}
\caption{Results from simulation study for the scenario in which partition
structure is included in data generation process.  Boxplots display the
100 WAIC values that correspond to model fit for each synthetic data
generating scenario.  Note that smaller indicates better fit.}
\label{waicDRP}
\end{center}
\end{figure}
\spacingset{1.5}

The primary purpose of Figure \ref{waicGP} is to compare model fit from
the spatio-temporal partition model we develop to that from the linear
DDPM  (model 5). It appears that all methods are competitive to the linear
DDPM, which is particularly true with 10 time points.  Thus,  our
dependent partition model accommodates temporal dependence more
efficiently relative to the linear DDPM under this data generating
scenario.  Note that regardless of the number of time points, model 1b
($stRPM(\alpha, \nu_0, M)$ with $\nu_0 = 20$) appears to perform best.
Surprisingly,  $tRPM(\alpha, M)$ (model 4) is quite competitive,
particularly with 10 time points.    The conclusion here is that employing
$stRPM(\alpha, \nu_0, M)$ to model partitions appears to accommodate
spatio-temporal dependence even if there is no underlying partition
structure.

From Figure \ref{waicDRP} we see that when partitions are generated
independently, there is very little lost by employing the dependent joint
model in terms of model fit (see top left panel for model 3 and 4).
However, as spatial and/or temporal structure is introduced in the
partition model, there are clear gains in terms of model fit when
employing $tRPM(\bm{\alpha}, M)$ and/or $stRPM(\bm{\alpha}, \nu_0,M)$.
From this simulation it seems that employing the $tRPM(\alpha, M)$
regardless of the strength of temporal dependence among partitions is
reasonable as there is minimal cost in terms of model fit even when
partitions are generated independently.  Finally, it appears that
$stRPM(\bm{\alpha}, \nu_0,M)$ performed best.

\section{Data Applications} \label{SIMCE}
In this section, we detail an additional application in the
field of education.

\subsection{SIMCE Data Application} \label{SIMCE.application}

Incorporating spatio-temporal structure in education studies has been
explored (e.g., \citealt{cuervo&anton:2013,
fotheringmah&chalrton&brundson:2001}). In school assessment and
effectiveness studies, temporal persistence in school performance is of
principal interest. It seems likely that school performance from
year-to-year is relatively stable except in circumstances where a school
undergoes  many changes in personnel (faculty and students) or curriculum
from one year to the next. In addition  it seems reasonable that
geographic location plays a role in school performance, particularly if
communities are segregated socio-economically which happens to be the case
in metropolitan area of Santiago, Chile. For these reasons we fit model
\eqref{FullModel} to these data as well.

In order to formally assess both national and school level education
effectiveness in Chile, the Chilean national learning outcome assessment
system (Sistema de Medici\'on de Calidad de la Educaci\'on, SIMCE) was
created to, among other things, administer standardized tests to education
institutions in Chile. Each year a standardized test in mathematics and
language is  administered to 4th grade students.
We were granted access to 7 years  of data (2005-2011) where the longitude
and latitude of most schools were recorded.

For the SIMCE data in addition to the 16 models fit to the PM$_{10}$ data,
we also considered an alternative to employing the sPPM at each time
period which is more computationally efficient. The alternative approach
models only $\rho_1$ with an $sPPM(\nu_0, M)$ (equation \eqref{sPPM} of
the main document) and the remaining $T-1$ partitions with an
$tRPM(\bm{\alpha}, M)$ (equation \eqref{joint.joint.model} of the main
document) with a $CRP(M)$ EPPF.  In this formulation, the strength of
$\alpha_t$  would be the only mechanism by which the spatial structure
found in $\rho_1$.

We employed the same prior parameter values here  as in Section
\ref{pm10.application} of the main document, except we set $A_{\sigma}=15$
and $\nu_0=2$ to account for the higher variability present in the SIMCE
data. Each of the 24 models were fit to the SIMCE data by collecting 1000
MCMC draws after discarding the first 5000 as burn-in and thinning by 5.
The LPML and WAIC results can be found in Table \ref{SIMCEfull}.

Similar to the PM$_{10}$ analysis, the best performing model in terms of
WAIC includes spatio-temporal dependence in the partition model, temporal
dependence among the atoms, and temporal dependence in the likelihood. The
best performing model in terms of LPML assumed atoms are $iid$.  Notice
further, that generally speaking, incorporating temporal dependence in the
model for $(\rho_1, \ldots, \rho_7)$ improves model fit.    It appears
that there is a cost in model fit associated with employing the
$sPPM(\nu_0, M)$ at the first time period and the $CRP(M)$ for subsequent
time periods in terms of model fit.  However, the cost is not exorbitant
relative to extraordinary computation gains (12 hours for model that
includes space at time 1 versus 6 days for the model that includes space
at each time point).   To see how estimated partitions from the two models
(that which includes space in the first time point versus that which
includes space at each time point) change over time, we provide Figure
\ref{time1vstime7}.  Notice that there is a change in dependence from time
period 2 and 3 and that the similarity between partitions decays when
including space at each time point. 

\begin{table}
\caption{Results of fitting 24 models to the SIMCE data.  The bold font
identifies the models that produced the best LPML and WAIC values. Higher
values for LPML indicate better fit while lower values for WAIC indicate
better fit.}
\begin{tabular}{c @{\hskip 0.05in} c @{\hskip 0.05in} c | cc cc cc }
& & & \multicolumn{6}{c}{Space} \\ \cmidrule(r){4-9}
& & & \multicolumn{2}{c}{No} & \multicolumn{4}{c}{Yes} \\ \cmidrule(r){4-5} \cmidrule(r){6-9}
\multicolumn{3}{c|}{Temporal Dependence In} & &  & \multicolumn{2}{c}{Each Time} & \multicolumn{2}{c}{First Time} \\\cmidrule(r){1-3}  \cmidrule(r){6-7} \cmidrule(r){8-9}
Partition & Likelihood & Atoms 	& LPML & WAIC & LPML & WAIC & LPML & WAIC \\ \midrule
No & No & No  					& -34094 & 62963 & -33543 & 62416 & -34054 & 62960 \\
No & No & Yes  				& -34040 & 62693 & -33558 & 62577 & -34044 & 63043 \\
No & Yes & No  				& -31214 & 60087 & -30701 & 60400 & -31129 & 59349 \\
No & Yes & Yes  				& -31241 & 59572 & -30712 & 60944 & -31115 & 59686 \\   \midrule
Yes & No & No  				& -32457 & 64835 & -30760 & 61045 & -31198 & 61516 \\
Yes & No & Yes  				& -31007 & 61948 & -31180 & 61348 & -32690 & 64578 \\
Yes & Yes & No  				& -30390 & 60340 & -$\bm{29936}$ & 58122 & -30573 & 60132 \\
Yes & Yes & Yes  				& -30378 & 60314 & -30959 & $\bm{57834}$ & -30331 & 59544 \\   \bottomrule
\end{tabular}
\label{SIMCEfull}
\end{table}%

\begin{figure}
\centering
\includegraphics[scale=0.5]{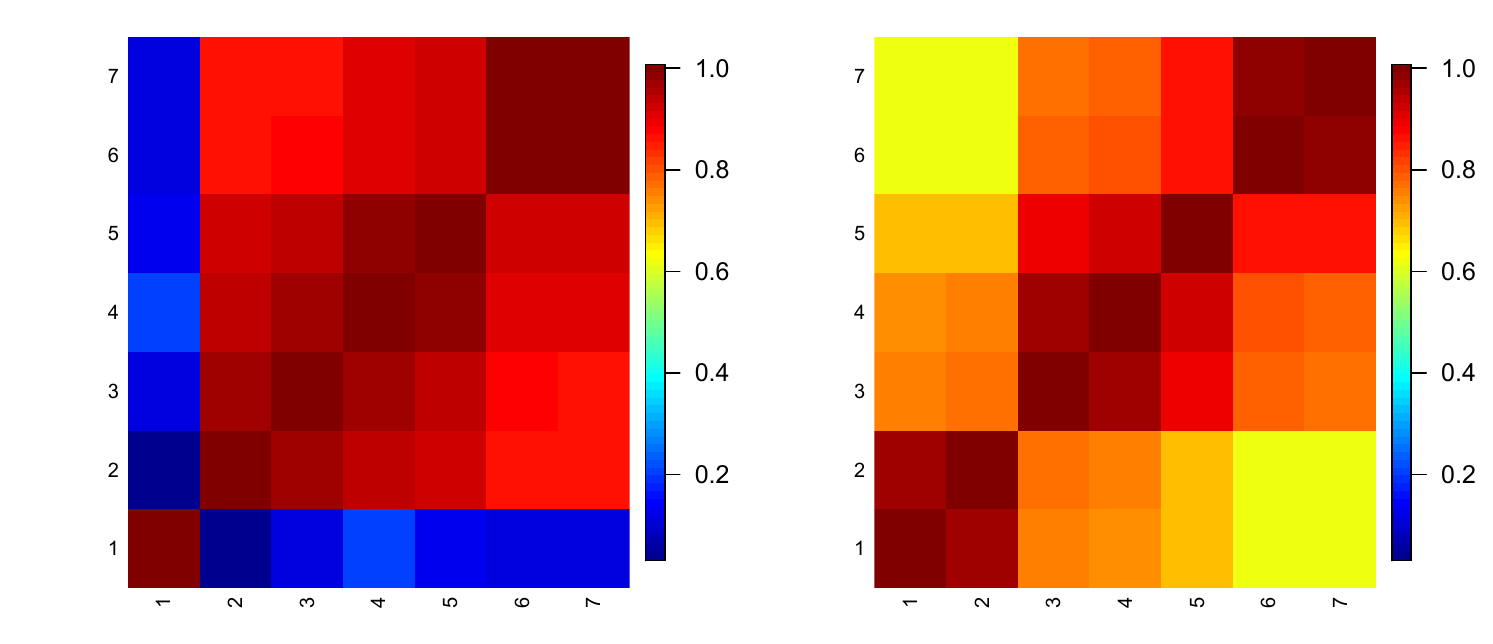}
\caption{Each figure is a summary of the lagged $ARI$ value corresponding to models
that include space in different ways.  The left plot corresponds to model
that includes space in partition model only at time period 1.  The right
plot corresponds to model that includes space in partition model at each
of the seven time periods.} \label{time1vstime7}
\end{figure}

\bibliographystyle{agsm}

\bibliography{reference}